\documentclass[]{article}
\usepackage{amssymb,amsmath,ascmac,amsthm,bm}
\usepackage[mathscr]{eucal}
\usepackage{mathrsfs,url,graphicx,multirow}
\def\<{\langle}
\def\>{\rangle}

\theoremstyle{plane}
\newtheorem{Thm}{Theorem}[section]
\newtheorem{Def}[Thm]{Definition}
\newtheorem{Prop}[Thm]{Proposition}
\newtheorem{Lem}[Thm]{Lemma}
\newtheorem{Cor}[Thm]{Corollary}

\newtheorem{Rem}[Thm]{Remark}

\newtheorem{Conj}[Thm]{Conjecture}

\newtheorem{Proc}[Thm]{Procedure}

\title{{\bf Formal Power Series on Algebraic Cryptanalysis}}
\author{Shuhei Nakamura \thanks{Department of Computer and Information Sciences, Ibaraki University, 4-12-1 Nakanarusawa-cho, Hitachi, Ibaraki 316-0033, Japan (E-mail: \texttt{shuhei.nakamura.fs71@vc.ibaraki.ac.jp})}}
\date{}
\begin{document}
\maketitle

\begin{abstract}
In the complexity estimation for an attack that reduces a cryptosystem to solving a system of polynomial equations, the degree of regularity and an upper bound of the first fall degree are often used in cryptanalysis. While the degree of regularity can be easily computed using a univariate formal power series under the semi-regularity assumption, determining an upper bound of the first fall degree requires investigating the concrete syzygies of an input system.
In this paper, we investigate an upper bound of the first fall degree for a polynomial system over a sufficiently large field.
In this case, we prove that the first fall degree of a non-semi-regular system is bounded above by the degree of regularity, and that the first fall degree of a multi-graded polynomial system is bounded above by a certain value determined from a multivariate formal power series.
Moreover, we provide a theoretical assumption for computing the first fall degree of a polynomial system over a sufficiently large field.
\end{abstract}

\section{Introduction}\label{sec:intr}
Solving the system of polynomial equations is an important topic in computer algebra since such a system can arise from many practical applications (cryptology, coding theory, computational game theory, optimization, etc.).
In particular, the so-called MQ problem that solves the system of quadratic equations is NP-hard even if over the finite field of order two \cite{Gar79}, and the security of multivariate cryptography \cite{DS} is based on the hardness of the problem.
Multivariate cryptography is expected to have a particularly high potential for building post-quantum signature schemes and is investigated in NIST post-quantum cryptography (PQC) standardization project \cite{NIST2016}.

In multivariate cryptography, there are several attacks reducing it to the problem solving the system of polynomial equations. 
The direct attack \cite{Bet09,FJ03} generates the system of quadratic equations from the public key and a ciphertext/message and obtains its plaintext/signature by solving the system.
Moreover, as a key recovery attack generating the system of polynomial equations, it is known such as the Rainbow-Band-Separation (RBS) attack \cite{DYCCC}, the MinRank attack \cite{Beu20,BG06,KS98}, Intersection attack \cite{Beu20}.
In particular, the MinRank attack reduces a cryptosystem to the MinRank problem and solves a MinRank instance by using a method generating the system of polynomial equations, such as the Kipnis-Shamir (KS) method \cite{KS98}, minors methods \cite{BBCGPSTV20,Fau08}.
The MinRank problem is also NP-hard \cite{SFB96} and is a base of the security of the rank metric code-based cryptography, which was proposed in NIST PQC standardization project \cite{NIST1R_PQC,NIST2R_ROLLO}.
The efficiency of these attacks depends on that of an algorithm, say {\it solver}, to solve the system of polynomial equations generated by an attack.

In cryptography, the solver with a Gr\"{o}bner basis or a (matrix) kernel search is mainly used as an efficient algorithm to solve the system generated by an attack.
The complexity of such a solver is given by using the {\it solving degree}, which is the maximal degree of polynomials required to solve the system.
However, the solving degree is an experimental value and it is practically hard to decide this value of a large-scale system arising from cryptography.
For the solver with a Gr\"{o}bner basis, the degree of regularity \cite{Bar04ICPSS} and the first fall degree \cite{DG10} are used as a tool approximating to the solving degree.
Also, for the solver with a kernel search, an analogue of the degree of regularity is used as such a tool under some assumptions.

For a semi-regular polynomial system, which generalizes a regular sequence in an overdetermined case, 
the degree of regularity introduced by M. Barded et al.\! \cite{Bar04ICPSS} is given by the minimum degree $D_{{\it reg}}$ of terms with a non-positive coefficient in the power series
\begin{equation}\label{eq:220325b}\frac{\prod _{i=1}^m(1-t^{d_i})}{(1-t)^n}, \end{equation}
where $d_1,\dots ,d_m$ and $n$ are the degrees and the number of the variables of the system, respectively.
Since the value $D_{{\it reg}}$ approximates the solving degree tightly and is simply computed by the power series \eqref{eq:220325b}, the degree of regularity is widely used for complexity estimation \cite{NIST2RLUOV,NIST2RGeMSS,NIST2RMQDSS,NIST2R}.
However, there exist non-semi-regular systems in cryptography, such as a quadratic system generated by the direct attack against HFE \cite{Pat96}, the RBS attack \cite{Nak20} and the MinRank attack using the KS method \cite{Nak20b}.

A non-semi-regular system appeared in cryptography is often solved within a smaller degree than the degree of regularity in the semi-regular case.
Namely, its solving degree is smaller than the degree of regularity. 
The first fall degree introduced by V. Dubois and N. Gama \cite{DG10} is defined by using non-trivial syzygies of the top homogeneous component and approximates the solving degree of a non-semi-regular system generated by some attacks. 
Since it is difficult to determine its actual value from large parameters in a cryptosystem, previous estimations \cite{DH11,DK11,DY13,Ver19} use an upper bound for the first fall degree.
These works provide this upper bound by investigating a concrete non-trivial syzygy and are important for providing an accurate cryptanalysis for a given cryptosystem.
However, the discussion for this upper bound depends on each cryptosystem and is not given to a general system.
Moreover, a relation with the degree of regularity is not clear theoretically.

In 2020, progress has been made in the complexity estimation for an attack generating a multi-graded polynomial system. For example, \cite{Nak20b} shows that the first fall degree of a multi-graded polynomial system generated by the MinRank attack using the KS method is bounded above by the minimal total degree $D_{{\it mgd}}$ of terms with a negative coefficient in 
\begin{equation}\label{eq:220325a}\frac{\prod _{i=1}^m(1-t_1^{d_1^{(i)}}t_2^{d_2^{(i)}}\cdots t_s^{d_s^{(i)}})}{(1-t_1)^{n_1}(1-t_2)^{n_2}\cdots (1-t_s)^{n_s}} \end{equation}
where $(d_1^{(i)},d_2^{(i)},\dots ,d_s^{(i)})$ is a multi-degree and $n_i$ is the size of each set of variables in the multi-graded system.
Moreover, in NIST PQC 2nd round, the estimations \cite{Nak20,SP20} in the bi-graded case influenced a parameter selection for Rainbow \cite{NIST3R}. 
However, the theoretical background for such estimations is not clear.

\subsection{Our contribution}
In this article, we mainly investigate an upper bound of the first fall degree for a polynomial system over a sufficiently large field. Upper bounds given in this article are actually that of the minimal degree $d_{{\it KSyz}}$ at the first homology of the Koszul complex to a polynomial system (Definition \ref{def:12}). The assumption for the order of a field gives that the minimal degree coincides with the first fall degree (Proposition \ref{prop:3}). \\
\indent Firstly, we prove that for a non-semi-regular system over a sufficiently large field, the degree of regularity is an upper bound of the first fall degree (Theorem \ref{thm:20230222c}). 
In general, we show that an upper bound of the first fall degree of a polynomial system is $D_{\mathbb{Z}_{\geq 0}}$ decided by the power series \eqref{eq:220325b} (Definition \ref{def:20230325a}, Theorem \ref{thm:20230221a}). If $D_{\it reg}=D_{\mathbb{Z}_{\geq 0}}$, then 
the first fall degree of a semi-regular system over a sufficient large field coincides with the degree of regularity (Corollary \ref{cor:20230325b}).\\
\indent Secondly, we naturally generalize the previous discussion to the multi-graded case. Namely, we define $D_{\mathbb{Z}_{\geq 0}^s}$ decided by the multivariate power series \eqref{eq:220325a} as a generalization of $D_{\mathbb{Z}_{\geq 0}}$ (Definition \ref{def:9}).
Then we prove that the first fall degree of a multi-graded polynomial system over a sufficiently large field is bounded above by $D_{\mathbb{Z}_{\geq 0}^s}$ (Theorem \ref{thm:8}).
In more detail, we also see that for a multi-graded polynomial system over a sufficiently large field, the multi-graded version of $d_{{\it KSyz}}$ is bounded above by ${\bf D}_{\mathbb{Z}_{\geq 0}^s,\prec}$ (Definition \ref{def:9}, \ref{def:123}, Theorem \ref{thm:8gen}).
These results theoretically guarantee recently developed cryptanalyses \cite{Nak20b,Nak20,SP20} using the first fall degree (\ref{ssec:app}). In particular, these cryptanalyses give examples such that $D_{\mathbb{Z}_{\geq 0}^s}$ is smaller than $D_{\mathbb{Z}_{\geq 0}}$.\\
\indent Finally, we provide an application using a generalization ${\bf D}_{\mathcal{X},\prec}^{\leq a}$ of ${\bf D}_{\mathbb{Z}_{\geq 0}^s,\prec}$, where $\mathcal{X}\subseteq \mathbb{Z}_{\geq 0}^s$ and $\prec $ is well-ordering on $\mathcal{X}$. We introduce a multi-graded XL algorithm with a (matrix) kernel search as a generalization of [4,46,48] and define a tool for estimating the complexity of this XL algorithm. Then we show that this tool is bounded from above by ${\bf D}_{\mathbb{Z}_{\geq a}^s,\prec}^{\leq a}$ under the assumption of a certain regularity (Proposition \ref{prop:0322c}).
In particular, this regularity provides computing the first fall degree of a polynomial system over a sufficiently large field (Proposition \ref{cor:20230417a}).

\subsection{Organization}
This article is organized as follows.
We explain some algorithmic backgrounds for our study in Section 2, and recall some fundamental concepts in commutative ring theory and proxies of the solving degree for estimating the complexity of a solver in Section 3.
In Section 4, we prove that the first fall degree is smaller than the degree of regularity in the semi-regular case if the order of the coefficient field is sufficiently large.
In Section 5, we prove that the first fall degree of a multi-graded polynomial system is bounded by a certain value determined from its multi-degree if the order of the coefficient field is sufficiently large, and provide the theoretical assumption for applying the XL algorithm with a kernel search.
In Section 6, we provide actual examples that satisfy the condition for the order of the coefficient field.

\subsection{Related work}
Bilinear polynomial systems can be seen as a special case of multi-graded polynomial systems. Among the studies on such systems, \cite{BCV20} investigated the relation between the first fall degree and the degree of regularity on the bi-degree $(1,d-1)$.
They explicitly determined the first fall degree $d_{{\bf y}\text{-}{\it ff}}$ and the degree of regularity $d_{{\bf y},{\it reg}}$ introduced by them under a certain regularity and clarified their relation.
In Subection \ref{ssec:20250914b}, we characterize their results by using formal power series, and we partially clarify the relation among $d_{{\bf y}\text{-}{\it ff}}$, $d_{{\bf y},{\it reg}}$ on the bi-degree $(1,d-1)$, as introduced by \cite{BCV20}, and $D_{\mathcal{X},\prec}^{\leq a}$ as introduced by us, where $\mathcal{X}=\{1\}\times \mathbb{Z}_{\geq 0}^s$ with the lexicographic order $\prec\;=\;\prec_{lex}$.
For example, for a bilinear system, we show that $d_{{\bf y}\text{-}{\it ff}}$ is bounded from above by $D_{\{1\}\times \mathbb{Z}_{\geq 0},\prec_{lex}}^{\leq -1}$ (Corollary \ref{cor:20250913b}). Moreover, under their regularity assumption, we see that $d_{{\bf y}\text{-}{\it ff}}=D_{\{1\}\times \mathbb{Z}_{\geq 0},\prec_{lex}}^{\leq -1}$ and $d_{{\bf y},{\it reg}}=D_{\{1\}\times \mathbb{Z}_{\geq 0},\prec_{lex}}^{\leq 0}$ (Corollary \ref{cor:20250914c}).

\section{Background}

In this section, we explain the complexity estimation for an algorithm to solve the system of polynomial equations.
The Gr\"{o}bner basis algorithm and the XL algorithm with a (matrix) kernel search are mainly used in cryptography such as the NIST PQC standardization project (see \cite{NIST2RLUOV,NIST2RGeMSS,NIST2RMQDSS,NIST2R}).
We explain the complexity estimation for the Gr\"{o}bner basis algorithm in Subsection \ref{ssec:solvergb}, for the XL algorithm with a kernel search in Subsection \ref{ssec:solverks}, and for a solver to a multi-graded polynomial system in Subsection \ref{ssec:solvermgd}.

\subsection{Solver with a Gr\"{o}bner basis}\label{ssec:solvergb}

For a given polynomial system $f_1,\dots ,f_m$, when the codimension of the ideal $\<f_1,\dots ,f_m\>$ is zero, the ideal $\<f_1,\dots ,f_m\>$ includes a univariate polynomial as a generator of the Gr\"{o}bner basis with respect to the lexicographic monomial order.
By solving this univariate polynomial, a variable in the solution to the system $f_1,\dots ,f_m$ is fixed.
Repeating this procedure, we obtain a solution of the system $f_1,\dots ,f_m$.

The Gr\"{o}bner basis algorithm, which computes the Gr\"{o}bner basis for a given ideal-generator, was discovered by Buchberger \cite{Buch65}, and faster algorithms are implemented such as F4 \cite{Fau99}, F5 \cite{Fau02} and XL family \cite{Yan07a}. 
The F4 algorithm, for example, is the default algorithm for computing Gr\"{o}bner bases in the computer algebra software Magma \cite{Magma}, and the complexity of that with respect to the graded reverse lexicographic monomial order is estimated as
\begin{equation}\label{eq:cmplF4dslv}{n+d_{{\it slv}}\choose d_{{\it slv}}}^\omega \end{equation}
where $n$ is the number of variables, $\omega $ is a linear algebra constant, and $d_{{\it slv}}$ is the solving degree that is the maximal degree of polynomials required to solve the system.
For example, $\omega =2$ as the direct attack against G$e$MSS \cite{NIST2RGeMSS}.
Moreover, as a conversion algorithm to the Gr\"{o}bner basis with respect to another monomial order, the FGLM algorithm returns to the Gr\"{o}bner basis with respect to the lexicographic monomial order in complexity 
\[nL^3\]
where $L$ is the number of solutions counted with multiplicity in the algebraic closure of the coefficient field \cite{FGLM93}.

In order to satisfy the conditions that the codimension of an ideal is zero and that the parameter $L$ is sufficiently small, some variables are fixed before solving a polynomial system.

\subsection{Solver with a kernel search}\label{ssec:solverks}

For simplicity, we assume that $f_1,\dots , f_m$ are homogeneous of degree two.
For a target degree $d \in \mathbb{Z}_{\geq 0}$, the XL (eXtended Linearization) algorithm generates a new system of degree $d$ by multiplying the system $f_1,\dots ,f_m$ by the monomials of degree $d-2$ and obtains the so-called {\it Macaulay matrix of degree $d$} whose row components correspond to the coefficients of a polynomial in the new system.
The columns of the Macaulay matrix correspond to the monomials of degree $d$.
We also call the {\it solving degree} of an XL algorithm the maximal degree of polynomials used in an XL algorithm.
Note that, when arranging the columns of the Macaulay matrix according to a monomial order, the XL algorithm obtains a new pivot (corresponding to a new {\it leading term}) by computing a row echelon form of the matrix and can return Gr\"{o}bner bases of degree $d$.
The XL algorithm is available as an algorithm computing the Gr\"{o}bner basis as explained in Subsection \ref{ssec:solvergb}. 

For a solution $(a_1,\dots ,a_n)\in \mathbb{F}^n$ to the system $f_1,\dots ,f_m$, the vector \[(\dots,  a_1^{d_1}\cdots a_n^{d_n},\dots )_{d_1+\cdots +d_n=d} \in \mathbb{F}^{{n+d-1\choose d}}\] 
is a kernel vector of the Macaulay matrix of degree $d$. The kernel space of the Macaulay matrix contains vectors corresponding to the solutions to the system. The {\it XL algorithm with a kernel search} solves the system by finding such a vector. 
When the algorithm randomly returns a kernel vector, its complexity is estimated as 
\begin{equation}\label{eq:cmplxlks}q^{c_d-u_d}\cdot C_{{\it KernelSearch}}\end{equation}
where $C_{{\it KernelSearch}}$ is the complexity of the kernel search, $c_d$ is the number of the columns minus the rank and $u_d$ is the dimension of the kernel space that corresponds the solution space of the system.
Here the term $q^{c_d-u_d}$ is an iteration required to find a kernel vector corresponding to a solution to the system.
Moreover, when an algorithm that returns a basis of the kernel space is used as a kernel search, the complexity is sufficient in $C_{{\it KernelSearch}}+q^{c_d-u_d}$.

If $c_d\leq 1$ holds, then $q^{c_d-u_d}=1$ and a kernel vector obtained by a kernel search always corresponds to a solution. 
For the XL algorithm with a kernel search to solve a homogeneous system, the following value is often used as a tool to approximate the solving degree: 
\begin{equation}\label{eq:degree0322}d^{\leq 1}=\min \{d>0\mid c_d\leq 1\}.\end{equation}
In Subsection \ref{ssec:estXL}, we explain the theoretical assumption for controlling value $c_d$.

The XL algorithm in \cite{YCBC07} extracts a square matrix from the obtained Macaulay matrix and uses the Wiedemann algorithm \cite{Wie86} as a kernel search of the matrix. 
The complexity of the Wiedemann algorithm is estimated as 
\[N\cdot k_{{\it total}} =N^2\cdot k\]
where $N$ is the number of columns, $k_{{\it total}}$ is the number of non-zero entries of the matrix and $k$ is the average number of non-zero entries in a row.
Since the XL algorithm multiplies each quadratic polynomial in the system by a monomial, the number of non-zero entries in each row of the obtained Macaulay matrix is at most the number of non-zero coefficients in the quadratic polynomial, namely 
\[k\leq {n+2\choose 2}.\]
The algorithm is called the Wiedemann XL algorithm and is used for the direct attack against Rainbow in \cite{NIST2R}.

\subsection{Solver to a multi-graded polynomial system}\label{ssec:solvermgd}

The XL algorithm with respect to a target degree $d$ generates the Macaulay matrix of a degree $d$ whose columns correspond to the monomials of degree $d$.
For a multi-graded polynomial system, we can naturally consider the XL algorithm that generates the Macaulay matrix whose columns correspond to the monomials of degree ${\bf d}\in \mathbb{Z}_{\geq 0}^s$. 
Then the number of columns is 
\[N_{{\bf n},{\bf d}}:={n_1+d_1-1\choose d_1}\cdots {n_s+d_s-1\choose d_s}\]
where ${\bf d}=(d_1,\dots ,d_s)$, and the Macaulay matrix contains a kernel vector corresponding to a solution to the multi-graded system.
In the XL algorithm with a kernel search, the complexity of the Wiedemann algorithm is estimated as 
\[N_{{\bf n},{\bf d}}^2\cdot k\]
where $k$ is the average number of non-zero entries in a row.


Moreover, in order to obtain the Gr\"{o}bner basis for a multi-graded polynomial system, the XL algorithms performed as follows (see Subsection \ref{ssec:fccrt} for notation).
\begin{Proc}\label{proc:20250913a}
Let $f_1,\dots ,f_m$ be a multi-graded polynomial system with multi-degrees $(d_1^{(1)},\dots ,d_s^{(1)}),\dots ,(d_1^{(m)},\dots ,$ $d_s^{(m)})$, respectively, and set ${\bf d}=(d_1,\dots ,d_s)\in \mathbb{Z}_{\geq 0}^s$.
\begin{enumerate}
\item Multiply each polynomial $f_i$ by the monomials of degree $(a_1,\dots ,a_s)$ such that $a_j+d_j^{(i)}\leq d_i$. Obtain a new multi-graded system.
\item Generate the Macaulay matrix corresponding to the resulting system and compute its echelon form.
\item Return the polynomials that correspond to the rows of the resulting matrix. 
\end{enumerate}
\end{Proc}
\noindent Note that this approach is a natural generalization from \cite{BCV20,SP20,Ver19}. 
When the Gaussian elimination is used as computing the echelon form, its complexity is estimated as 
\[R\cdot \tilde{N}_{{\bf n},{\bf d}}^{\omega -1}\]
where $\omega $ is a linear algebra constant, $\tilde{N}_{{\bf n},{\bf d}}={n_1+d_1\choose d_1}$ $\cdots {n_s+d_s\choose d_s}$ and $R=\sum _{i=1}^m\sum _{{\bf a}+\deg f_i\preceq {\bf d}}={n_1+a_1\choose a_1}\cdots {n_s+a_s\choose a_s}$. 
Here, for ${\bf a}=(a_1,\dots ,a_s),{\bf b}=(b_1,\dots ,$ $b_s)\in \mathbb{Z}_{\geq 0}^s$, the relation ${\bf a}\preceq {\bf b}$ is defined as $a_i\leq b_i$ for each $1\leq i\leq s$.

\section{Preliminaries}\label{sec:prel}
In this section, we recall some fundamental concepts in commutative ring theory and the complexity estimation for a Gr\"{o}bner basis algorithm.

\subsection{Fundamental concepts in commutative ring theory}\label{ssec:fccrt}
A (Noetherian) commutative ring $R$ is said to be {\it $\mathbb{Z}^s$-graded} if it has a decomposition $R=\bigoplus _{{\bf d}\in \mathbb{Z}^s}R_{{\bf d}}$ such that $R_{{\bf d}_1}R_{{\bf d}_2}\subseteq R_{{\bf d}_1+{\bf d}_2}$ for any ${\bf d}_i\in \mathbb{Z}^s$.
In this article, we assume always that $R_{{\bf d}}=\{0\}$ if ${\bf d}\in \mathbb{Z}^s$ has a negative component, and that $\dim R_{\bf d}<\infty $ for every ${\bf d}\in \mathbb{Z}^s$.
Then we call $R$ a {\it $\mathbb{Z}_{\geq 0}^s$-graded commutative ring}. 
An element $h$ in $R_{\bf d}$ is said to be {\it $\mathbb{Z}_{\geq 0}^s$-homogeneous}, or simply {\it homogeneous} and then we denote ${\bf d}$ by $\deg _{\mathbb{Z}_{\geq 0}^s}h$ and call it {\it the $\mathbb{Z}_{\geq 0}^s$-degree of $h$}.
If $R_{{\bf 0}}$ is a field, put \[{\rm HS}_R({\bf t})=\sum _{{\bf d}\in \mathbb{Z}_{\geq 0}^s}(\dim _{R_{{\bf 0}}}R_{{\bf d}})\cdot {\bf t}^{{\bf d}}\in \mathbb{Z}_{\geq 0}[[t_1,\dots ,t_s]],\] where ${\bf d}=(d_1,\dots ,d_s)$ and ${\bf t}^{{\bf d}}=t_1^{d_1}\cdots t_s^{d_s}$ which is called a {\it (multivariate) Hilbert series}. 
For a $\mathbb{Z}_{\geq 0}^s$-graded commutative ring $R$, its quotient ring with an ideal generated by $\mathbb{Z}_{\geq 0}^s$-homogeneous elements is $\mathbb{Z}_{\geq 0}^s$-graded. 
For example, the polynomial ring $\mathbb{F}[{\bf x}_1,\dots ,{\bf x}_s]$ is $\mathbb{Z}_{\geq 0}^s$-graded by $\deg _{\mathbb{Z}_{\geq 0}^s}x_{ij}=(0,\dots ,0,\overset{i}{1},0,\dots ,0),$
where ${\bf x}_i=\{x_{i1},\dots ,x_{in_i}\}$.
Then we have a decomposition $\mathbb{F}[{\bf x}_1,\dots ,{\bf x}_s]=\bigoplus _{{\bf d}\in \mathbb{Z}_{\geq 0}^s}\mathbb{F}[{\bf x}_1,\dots ,{\bf x}_s]_{{\bf d}}$ where $\mathbb{F}[{\bf x}_1,\dots ,{\bf x}_s]_{{\bf d}}$ is the vector space over $\mathbb{F}$ generated by the monomials of $\mathbb{Z}_{\geq 0}^s$-degree ${\bf d}$.
When $s=1$, for a polynomial $f$ in $\mathbb{F}[x_1,\dots ,x_n]$ and a well-ordering $<$ on $\mathbb{Z}_{\geq 0}$, we consider an expression $f=\sum _d f_d$ where $f_d\in \mathbb{F}[x_1,\dots ,x_n]_d$ and call its homogeneous component of degree $\max _<\{d\mid f_d\not =0\}$ {\it its top homogeneous component}.
When $\deg _{\mathbb{Z}_{\geq 0}}x_i=1$, the polynomial ring is said to be {\it standard graded} and we denote $\deg _{\mathbb{Z}_{\geq 0}}f$ as $\deg f$.
In this article, we mainly treat a homogeneous case and call an element of $\mathbb{F}[x_1,\dots ,x_n]^m$ a {\it system}. 
Moreover, we call a system whose components are homogeneous a {\it homogeneous system}.

Let $R$ be a commutative ring.
For $(h_1,\dots ,h_m)\in R^m$, we define an $R$-module homomorphism 
\[R^m\rightarrow R, (b_1,\dots ,b_m)\mapsto \sum _{i=1}^mb_ih_i.\]
Then we denote by ${\rm Syz}_R(h_1,\dots ,h_m)$, or simply ${\rm Syz}(h_1,\dots, h_m)$, the kernel of this homomorphism and its element is called a {\it syzygy} of $(h_1,\dots ,h_m)$.
For example, 
\[\pi _{ij}:=(0,\dots ,0,\underset{i}{-h_j},0,\dots ,0,\underset{j}{h_i},0,\dots ,0),\]
where $1\leq i<j\leq m$, is such an element.
Here, we denote by ${\rm KSyz}(h_1,\dots ,h_m)$ the submodule generated by the elements $\pi _{ij}$ and its element is called a {\it Koszul syzygy}.
Let $R$ be a $\mathbb{Z}_{\geq 0}^s$-graded commutative ring.
For $\mathbb{Z}_{\geq 0}^s $-homogeneous elements $h_1,\dots ,h_m\in R$ with ${\bf d}_i=\deg _{\mathbb{Z}_{\geq 0}^s}h_i$ and the free module $E=R^m$ with the standard basis $\{{\bf e}_1,\dots ,{\bf e}_m\}$, we denote by $K_i$ the $i$-th exterior power $\bigwedge ^iE$ and give the {\it Koszul complex}
\[K(h_1,\dots ,h_m)_\bullet :\cdots \overset{\delta _4}{\rightarrow }K_3\overset{\delta_3}{\rightarrow} K_2\overset{\delta_2}{\rightarrow } K_1 \overset{\delta_1}{\rightarrow }R\rightarrow 0\]
by $\delta _l({\bf e}_{i_1}\wedge \cdots \wedge {\bf e}_{i_l})=\sum _{k=1}^l(-1)^{k+1}h_{i_k}{\bf e}_{i_1}\wedge \cdots \wedge {\bf e}_{i_{k-1}}\wedge {\bf e}_{i_{k+1}}\wedge \cdots \wedge {\bf e}_{i_l}.$
Then we denote by $H_i(K(h_1,\dots ,h_m)_\bullet )$ the $i$-th homology group of the complex and by $H_i(K(h_1,\dots ,h_m)_\bullet )_{{\bf d}}$ its component of degree ${\bf d}$.
Note that 
\[H_1(K(h_1,\dots ,h_m)_\bullet )_{{\bf d}}={\rm Syz}(h_1,\dots ,h_m)_{{\bf d}}/{\rm KSyz}(h_1,\dots ,h_m)_{{\bf d}}.\]
Define $K(h_1,\dots ,h_s)(-{\bf d})$ as $K(h_1,\dots ,h_s)(-{\bf d})_{{\bf d}_0}:=K(h_1,\dots ,h_s)_{{\bf d}_0-{\bf d}}$ for any ${\bf d}_0$.
Since $K(h_1,\dots ,h_m)_\bullet $ is the mapping cone of \[K(h_1,\dots ,h_{m-1})_\bullet (-{\bf d}_m)\overset{\times h_m}{\rightarrow }K(h_1,\dots ,h_{m-1})_\bullet ,\] we have the short exact sequence (see \cite{Die15,Sch03}):
\[0\rightarrow K(h_1,\dots ,h_{m-1})_\bullet \rightarrow K(h_1,\dots ,h_m)_\bullet \rightarrow K(h_1,\dots ,h_{m-1})_\bullet (-{\bf d}_m)[-1]\rightarrow 0.\]
Namely, we can obtain the following long exact sequence:
\begin{equation}\label{eq:les}
\hspace{-10pt}\begin{array}{l}
\cdots \rightarrow H_2(K(h_1,\dots ,h_{m-1})_\bullet )(-{\bf d}_m)\\
\overset{\times h_m}{\rightarrow }H_2(K(h_1,\dots ,h_{m-1})_\bullet ) \rightarrow H_2(K(h_1,\dots ,h_m)_\bullet )\\
\rightarrow H_1(K(h_1,\dots ,h_{m-1})_\bullet )(-{\bf d}_m)\\
\overset{\times h_m}{\rightarrow }H_1(K(h_1,\dots ,h_{m-1})_\bullet ) \rightarrow H_1(K(h_1,\dots ,h_m)_\bullet )\\
\rightarrow R/\langle h_1,\dots ,h_{m-1}\rangle (-{\bf d}_m)\overset{\times h_m}{\rightarrow}R/\langle h_1,\dots ,h_{m-1}\rangle \\
\rightarrow R/\langle h_1,\dots ,h_m\rangle \rightarrow 0
\end{array}
\end{equation}

\subsection{Complexity estimation for a Gr\"{o}bner basis algorithm}\label{ssec:2.2}
In this subsection, we treat only the standard graded case (see Subsection \ref{ssec:fccrt}).

In Subsection \ref{ssec:solvergb}, we explain that the complexity of the Gr\"{o}bner basis algorithm depends on the solving degree $d_{{\it slv}}$ (see Equation \eqref{eq:cmplF4dslv} for example).
However, the solving degree is an experimental value.
In order to estimate the complexity of solving a large-scale polynomial system appearing in cryptography, we need to consider a theoretical tool to approximate the solving degree.  

The degree of regularity introduced by Bardet et al.\! \cite{Bar04ICPSS} is well-known as such a tool. 
\begin{Def}[\cite{Bar04ICPSS}]
For polynomials $f_1,\dots ,f_m$ in the (standard graded) polynomial ring $\mathbb{F}[x_1,\dots ,x_n]$ with $d_i=\deg f_i$, we define the {\rm degree of regularity} $d_{{\it reg}}$ as follows:
\[d_{{\it reg}}(f_1,\dots ,f_m)=\inf \{d\mid \langle f_1^{top},\dots ,f_m^{top}\rangle _d=\mathbb{F}[x_1,\dots ,x_n]_d\}\cup \{\infty \},\]
where $f_i^{top}$ is its top homogeneous component of $f_i$.
\end{Def}

When the top homogeneous component of the system $(f_1,\dots ,f_m)$ is {\it semi-regular} \cite{Bar04ICPSS} (also see Remark \ref{rem:4}), the degree of regularity coincides with the degree $D_{{\it reg}}$ defined as follows:

\begin{Def}
For the system $f_1,\dots, f_m$ with $d_i=\deg f_i$, we define $D_{{\it reg}}=D_{{\it reg}}(f_1,\dots ,f_m)$ as the minimum degree of terms with a non-positive coefficient in the power series 
\begin{equation}\label{eq:20220407a}\frac{\prod _{i=1}^m(1-t^{d_i})}{(1-t)^n}.\end{equation}
\end{Def}

In \cite{Fro85}, Fr\"{o}berg's conjectured that the Hilbert series for the quotient ring of a ``generic system" is determined by the power series \eqref{eq:20220407a} when the characteristic of the coefficient field is zero.
Indeed, this conjecture holds under the some conditions (see \cite{Ani86,Fro85,Sta80}).
Moreover, Diem in \cite{Die04} proved that, for any polynomial system that has the same number of variables and the degrees as the generic system, each coefficient of the Hilbert series for its quotient ring is larger than that for the generic system. 
Hence Fr\"{o}berg's conjecture induces the following statement: 
\begin{Conj}
Let $\mathbb{F}$ be a field of characteristic zero. 
For any polynomials $g_1,\dots ,g_m$ in $\mathbb{F}[x_1,\dots ,x_n]$ such that $\deg g_i=d_i$, the following inequality holds:
\[D_{\it reg}\leq d_{\it reg}(g_1,\dots ,g_m).\]
\end{Conj}
For a semi-regular system, the degree of regularity $D_{{\it reg}}$ tightly approximates the solving degree.
However, for a non-semi-regular system such as a multi-graded polynomial system, it is known that its Gr\"{o}bner basis is computed within a smaller degree than the degree of regularity (see Example \ref{rem:20251224a} or \cite{Nak20,Nak20b}).
By the conjecture above, we do not expect to use the degree of regularity as a tool for the solving degree of a non-semi-regular system.

After the works \cite{FJ03,GJS06} that observe the Gr\"{o}bner basis algorithm on a non-semi-regular system  generated by the direct attack against HFE \cite{Pat96}, 
Dubois and Gama \cite{DG10} propose the {\it first fall degree} as a theoretical tool for approximating the solving degree of a non-semi-regular system, and later Ding and Hodges \cite{DH11} reformulate the definition into an essential form (see Definition \ref{def:ffd}). Note that the first fall degree does not always work, as it is possible to intentionally create a non-trivial syzygy \cite{DS13}.

In order to capture a non-trivial degree fall occurs on the system, the {\it first fall degree} is defined using the non-trivial syzygies of the top homogeneous component for the polynomial system as follows: Let $B=\mathbb{F}_q[x_1,\dots ,x_n]/\langle x_1^q,\dots ,x_n^q\rangle $ with the standard graded decomposition $B=\oplus _{d\geq 0}B_d$.
For any $h\in \mathbb{F}[x_1,\dots ,x_n]$, we denote by $\overline{h}$ the image of $h$ under the natural surjection $\mathbb{F}_q[x_1,\dots ,x_n]\rightarrow \mathbb{F}_q[x_1,\dots ,x_n]/\langle x_1^q,\dots ,x_n^q\rangle $.
Let $d_0$ be a positive integer and $h_1,\dots ,h_m\in \mathbb{F}_q[x_1,\dots ,x_n]_{d_0}$.
For a positive integer $d$, we consider 
\[\varphi _d:B_{d-d_0}^m\rightarrow B_d,(\overline{b}_1,\dots ,\overline{b}_m)\mapsto \overline{b}_1\overline{h}_1+\cdots +\overline{b}_m\overline{h}_m.\]
Then, for 
\[\overline{\pi} _{ij}:=(0,\dots ,0,-\overline{h}_j,0,\dots ,0,\overline{h}_i,0,\dots ,0),\quad \]\[\overline{\tau} _i:=(0,\dots ,0,\overline{h}_i^{q-1},0,\dots ,0),\]
we define the following subspace of ${\rm Syz}_B(\overline{h}_1,\dots ,\overline{h}_m)_d={\rm Ker}\,\varphi _d$:
\[{\rm TSyz}_B(\overline{h}_1,\dots ,\overline{h}_m)_d:=\langle \overline{b}_{ij}\overline{\pi} _{ij},\overline{b}_i\overline{\tau} _i\mid \overline{b}_{ij}\in B_{d-2d_0},\overline{b}_i\in B_{d-d_0q}\rangle _{\mathbb{F}_q}. \]

\begin{Def}\label{def:ffd}
For any $f_1,\dots ,f_m\in \mathbb{F}[x_1,\dots ,x_n]$ such that $\deg f_i=d_0$, the first fall degree $d_{{\it ff}}(f_1,\dots ,f_m)$ is defined as
\[d_{{\it ff}}(f_1,\dots ,f_m)=\inf \{d\mid {\rm Syz}_B(\overline{f_1^{top}},\dots ,\overline{f^{top}_m})_d\not ={\rm TSyz}_B(\overline{f_1^{top}},\dots ,\overline{f^{top}_m})_d\}\cup \{\infty \}.\]
\end{Def}

There are several works on the first fall degree.
For example, as an analysis for the direct attack against HFE \cite{Pat96}, Ding and Hodges \cite{DH11} give an upper bound for the first fall degree by constructing a non-trivial syzygy of the generated quadratic system and prove that it can be solved in a quasi logarithmic time.
Moreover, as an analysis for the KS method \cite{KS98} for the MinRank problem \cite{Fau08}, Verbel et al.\! \cite{Ver19} investigate a non-trivial syzygy (of the top homogeneous component) of the generated quadratic system and give a new complexity estimation with the first fall degree for this method.
However, the previous evaluations for the first fall degree depend on each cryptosystem, and it is hard to apply to others.

For a given polynomial system, the degree of regularity and the first fall degree is defined on its top homogeneous component. 
It suffices to show a homogeneous case for discussing them.
Thus, in this article, we always assume that a polynomial is homogeneous with respect to the standard grading.

\begin{Rem}\label{rem:20251224a}
{\rm For multi-graded polynomial systems, it is known that there exist cases where the solving degree is smaller than that of a semi-regular system of the same size. In particular, such a multi-graded polynomial system is not semi-regular. 
Table \ref{table:20251223a} compares semi-regular and bi-graded systems of $9$ polynomials
in $\mathbb{F}[\bf x,\bf y]$ with $|\bf x|=5$ and $|\bf y|=3$, and shows that such cases indeed occur.
As defined in Definition \ref{def:9}, the second line $D_{\mathbb{Z}_{\geq 0}^2}$ in Table~\ref{table:20251223a} is the minimum degree among all terms $t_1^{d_1}t_2^{d_2}$ with a negative coefficient in the two-variable series
$(1-t_1^2)^{m_{(2,0)}}(1-t_1t_2)^{m_{(1,1)}}(1-t_2^2)^{m_{(0,2)}}(1-t_1)^{-5}(1-t_2)^{-3}$.
The third line ${\bf D}_{\mathbb{Z}_{\geq 0}^2,\prec}$ denotes the exponent $(d_1,d_2)$ of the smallest such term with respect to the graded reverse lexicographic order $\prec$.}
\end{Rem}
\begin{table}[h]
\caption{Comparison of semi-regular and randomly generated bi-graded systems of $9$ polynomials in $\mathbb{F}[{\bf x},{\bf y}]$ with $|{\bf x}|=5$ and $|{\bf y}|=3$. 
Here, $m_{(i,j)}$ denotes the number of polynomials of bi-degree $(i,j)$ and $D_{\mathbb{Z}_{\geq 0}^2}$ and ${\bf D}_{\mathbb{Z}_{\geq 0}^2,\prec}$ are defined in Definition \ref{def:9}.}
\label{table:20251223a}
\begin{center}
\begin{tabular}{|c|c|c|c|c||c|} \hline
&\multicolumn{4}{c||}{$(m_{(2,0)},m_{(1,1)},m_{(0,2)})$}&\multirow{2}{*}{Semi-reg.}\\
&$(1,4,4)$&$(2,4,3)$&$(3,4,2)$&$(4,4,1)$&\\ \hline
$d_{{\it slv}}$&3&4&4&5&$d_{{\it reg}}=5$\\ 
$D_{\mathbb{Z}_{\geq 0}^2}$&3&4&4&5&-\\
${\bf D}_{\mathbb{Z}_{\geq 0}^2,\prec}$&(0,3)&(1,3)&(2,2)&(2, 3)&-\\ \hline
\end{tabular}
\end{center}
\end{table}
In this article, we provide tools for estimating the solving degree of multi-graded polynomial systems.
In Section \ref{sec:multi}, we will prove that the first fall degree is bounded from above by $D_{\mathbb{Z}_{\geq 0}^2}$.

\section{Estimation using the standard grading}\label{sec:stand}
In Subsection \ref{ssec:proxyks}, we introduce a certain tool for approximating the solving degree which coincides with the first fall degree of a polynomial system over a sufficiently large field. In Subsection \ref{ssec:diem2015}, we prove that the first fall degree of a non-semi-regular system over such a field is smaller than the degree of regularity in the semi-regular case.

\subsection{Theoretical tool $d_{{\it KSyz}}$}\label{ssec:proxyks}
In this section, for a polynomial in $\mathbb{F}[x_1,\dots ,x_n]$, we only use the standard grading, i.e. $\deg _{\mathbb{Z}_{\geq 0}}x_i=1$.

We introduce the following definition for computing the first fall degree.

\begin{Def}\label{def:12}
For homogeneous polynomials $h_1,\dots ,h_m$ in the polynomial ring $\mathbb{F}[x_1,\dots ,x_n]$, we define $d_{{\it KSyz}}$ as 
\[\hspace{-15pt}d_{{\it KSyz}}(h_1,\dots ,h_m)=\inf \{d\mid {\rm Syz}(h_1,\dots ,h_m)_{d}\not ={\rm KSyz}(h_1,\dots ,h_m)_{d}\}\cup \{\infty \}.\]
\end{Def}

Let $\pi :\mathbb{F}_q[x_1,\dots ,x_n]\rightarrow \mathbb{F}_q[x_1,\dots ,x_n]/\langle x_1^q,\dots ,x_n^q\rangle $.
By the following lemmas, we see that $d_{{\it KSyz}}$ coincides with the first fall degree $d_{{\it ff}}$ for a sufficiently large $q$.

\begin{Lem}\label{lem:1}
For homogeneous polynomials $h_1,\dots ,h_m\in \mathbb{F}_q[x_1,\dots ,x_n]$ such that $\deg h_i=d_0$, if $q>d_{{\it ff}}(h_1,\dots ,h_m)$, then we have $d_{{\it ff}}(h_1,\dots ,h_m)\geq d_{{\it KSyz}}(h_1,\dots ,h_m)$.
\end{Lem}
\begin{proof}
 Let $B=\mathbb{F}_q[x_1,\dots ,x_n]/\langle x_1^q,\dots ,x_n^q\rangle $. 
Put $d=d_{{\it ff}}(h_1,\dots ,h_m)$.
There exists $\overline{\rho} =(\overline{\rho} _1,\dots ,\overline{\rho }_m)\in B_{d-d_0}^m$ such that 
\[\overline{\rho} \in {\rm Syz}_B(\overline{h}_1,\dots ,\overline{h}_m)_d\setminus {\rm TSyz}_B(\overline{h}_1,\dots ,\overline{h}_m)_d,\]
where $\deg \rho _i=d-d_0$ holds as a representative.
Then, since $\sum _{i=1}^m\overline{\rho}_i\overline{h}_i=0\in B_d$, we have $\rho _1h_1+\cdots +\rho _mh_m\in \langle x_1^q,\dots ,x_n^q\rangle $.
Since $\deg \rho _i+\deg h_i=d<q$, it follows that $\rho _1h_1+\cdots +\rho _mh_m=0$.
If $\rho \in {\rm KSyz}_R(h_1,\dots ,h_m)$, we obtain a contradiction $\overline{\rho }\in {\rm TSyz}_B(h_1,\dots ,h_m)$. 
Thus $\rho \in {\rm Syz}_R(h_1,\dots ,h_m)\setminus {\rm KSyz}_R(h_1,\dots ,h_m)$. Therefore, we have $d\geq d_{{\it KSyz}}(h_1,\dots ,h_m)$.
\end{proof}

\begin{Lem}\label{lem:2}
For homogeneous polynomials $h_1,\dots ,h_m\in \mathbb{F}_q[x_1,\dots ,x_n]$ such that $\deg h_i=d_0\geq 2$, if $q>d_{{\it KSyz}}(h_1,\dots ,h_m)$, then we have $d_{{\it ff}}(h_1,\dots ,h_m)\leq d_{{\it KSyz}}(h_1,\dots ,h_m)$.
\end{Lem}
\begin{proof}
Let $B=\mathbb{F}_q[x_1,\dots ,x_n]/\langle x_1^q,\dots ,x_n^q\rangle $. Put $d=d_{{\it KSyz}}(h_1,\dots ,h_m)$.
There exists $\rho \in {\rm Syz}_R(h_1,\dots ,h_m)_d\setminus {\rm KSyz}_R(h_1,\dots ,h_m)_d$.
In particular, $\overline{\rho }\in {\rm Syz}_B(\overline{h_1},\dots $ $,\overline{h_m})_d$.
Although an element $\overline{b}_i\overline{\tau} _i$ for $\overline{b}_i\in B_{d-d_0q}$ is contained in ${\rm TSyz}_B(\overline{h}_1,\dots ,\overline{h}_m)_d$ as generators, they do not appear since $B_{d-d_0q}=0$ by $d_0\geq 2$ and $q>d$.
Thus, if there exists $\overline{\rho} \in {\rm TSyz}_B(\overline{h}_1,\dots ,\overline{h}_m)_d$, then $\overline{\rho }=\sum _{i,j}\overline{b}_{ij}\overline{\pi }_{ij}$ for some $\overline{b}_{ij}$ where $\deg b_{ij}=d-d_0$ as a representative. 
Namely, $\rho -\sum _{i,j}b_{ij}\pi _{ij}\in \langle x_1^q,\dots ,x_n^q\rangle ^m$.
By $d<q$, we have $\rho -\sum _{i,j}b_{ij}\pi _{ij}=(0,\dots ,0)$.
Thus, we obtain a contradiction $\rho =\sum _{i,j}b_{ij}\pi _{ij}\in {\rm KSyz}(h_1,\dots ,h_m)_d$.
Therefore, $d_{{\it ff}}(h_1,\dots ,h_m)\leq d$. 
\end{proof}

\begin{Prop}\label{prop:3}
For homogeneous polynomials $h_1,\dots ,h_m$ $\in \mathbb{F}_q[x_1,\dots ,x_n]$ such that $\deg h_i=d_0\geq 2$, if $q>\min \{d_{{\it ff}}(h_1,\dots ,h_m),d_{{\it KSyz}}(h_1,\dots ,h_m)\}$, then we have $d_{{\it ff}}(h_1,\dots ,h_m)=d_{{\it KSyz}}(h_1,\dots ,h_m)$.
\end{Prop}
\begin{proof}
When $q>d_{{\it ff}}$, we have $d_{{\it ff}}\geq d_{{\it KSyz}}$ by Lemma $\ref{lem:1}$. 
Then $q>d_{{\it ff}}\geq d_{{\it KSyz}}$, and we have $d_{{\it ff}}\leq d_{{\it KSyz}}$ by Lemma $\ref{lem:2}$.
Thus $d_{{\it ff}}=d_{{\it KSyz}}$.
Similarly, when $q>d_{{\it KSyz}}$, we have the same result. 
\end{proof}

In \ref{ssec:app}, we show that the assumption in this proposition, i.e. $q>\min \{d_{{\it ff}}(h_1,\dots $ $,h_m),d_{{\it KSyz}}(h_1,\dots ,h_m)\}$, is satisfied in attacks against actual cryptosystems.
In the next subsection, under the assumption, we see that the first fall degree of a non-semi-regular system is smaller than the degree of regularity in the semi-regular case.

\subsection{Diem's work}\label{ssec:diem2015}
In \cite{Die15}, C. Diem investigates a relation between the regularity of a polynomial system and its syzygies. 

The following definition is introduced in \cite{Die15}.

\begin{Def}[\cite{Die15}]\label{def:regupto}
A homogeneous system $h_1,\dots ,h_m \in \mathbb{F}[x_1,\dots ,x_n]$ is {\rm regular up to degree $d$} if the following multiplication map $\times h_i$ by $h_i$ is injective for each $i=1,\dots ,m$.
\[\hspace{-10pt}\times h_i:(S/\langle h_1,\dots ,h_{i-1}\rangle )_{d_0-\deg h_i}\rightarrow (S/\langle h_1,\dots ,h_{i-1}\rangle )_{d_0},(\forall d_0\leq d)\]
\end{Def}

\begin{Rem}\label{rem:4}
{\rm 
In their article \cite{Bar04ICPSS}, Bardet et al.\! define a {\it semi-regular} system for a homogeneous system.
Note that a homogeneous system is semi-regular if and only if is regular up to degree $d_{{\it reg}}-1$ (see \cite{Die15}).
They also call a polynomial system semi-regular when its top homogeneous component is semi-regular.}
\end{Rem}

\begin{Def}
Let $\preceq $ be a well-ordering on $\mathbb{Z}_{\geq 0}^s$. 
\begin{enumerate}
\item For two elements $a,b$ of the formal power series ring $\mathbb{Z}[[t_1,\dots ,t_s]]$, we denote $a\equiv _{\preceq {\bf d}}b$ if the coefficients of these monomials of degree less than or equal to ${\bf d}$ with respect to $\preceq $ are the same.
\item For a $\mathbb{Z}_{\geq 0}^s$-graded module $M$ over a $\mathbb{Z}_{\geq 0}^s$-graded commutative ring $R$, i.e. it has $M=\bigoplus _{{\bf d}\in \mathbb{Z}_{\geq 0}^s}M_{{\bf d}}$ such that $R_{{\bf d}_1}M_{{\bf d}_2}\subseteq M_{{\bf d}_1+{\bf d}_2}$ for any ${\bf d}_i$, we put $M_{\preceq {\bf d}}=\bigoplus _{{\bf d}_0\preceq {\bf d}}M_{{\bf d}_0}$.
\end{enumerate}
\end{Def}

The regularity in Definition \ref{def:regupto} is characterized as follows:

\begin{Prop}[\cite{Die15}]\label{prop:Die15}
Let $S$ be the standard graded polynomial ring $\mathbb{F}[x_1,\dots ,x_n]$ with $\mathbb{F}[x_1,\dots ,x_n]=\bigoplus _{d\in \mathbb{Z}_{\geq 0}}\mathbb{F}[x_1,\dots ,x_n]_{d}$. 
For a homogeneous system $h_1,\dots ,h_m\in S$ with $\deg h_i \not=0$, the following conditions are equivalent:
\begin{enumerate}
\item $h_1,\dots ,h_m$ is regular up to degree $d$
\item ${\rm HS}_{S/\langle h_1,\dots ,h_m\rangle }(t)\equiv_{\leq d}\prod _{i=1}^m(1-t^{\deg h_i}) {\rm HS}_S(t)$
\item $H_1(K(h_1,\dots ,h_m)_\bullet )_{\leq d}=\{0\}$
\end{enumerate}
\end{Prop}

Since \[H_1(K(h_1,\dots ,h_m)_\bullet )_{d}={\rm Syz}(h_1,\dots ,h_m)_d/{\rm KSyz}(h_1,\dots ,h_m)_d,\] for the value $d_{{\it KSyz}}$ in Definition \ref{def:12}, note that we have 
\[d_{{\it KSyz}}(f_1,\dots ,f_m)=\inf \{d\mid H_1(K(f^{top}_1,\dots ,f^{top}_m)_\bullet )_{d}\not =0\}\cup \{\infty \}.\]

For a non-semi-regular system over a sufficiently large field, we see a relation between the first fall degree and the degree of regularity:
\begin{Thm}\label{thm:20230222c}
For a non-semi-regular system $h_1,\dots, h_m\in \mathbb{F}_q[x_1,\dots, x_n]$, we have
\[d_{{\it KSyz}}(h_1,\dots ,h_m)+1\leq  d_{\it reg}(h_1,\dots ,h_m).\] 
Moreover, if $\deg h_i=d_0\geq 1$ and $q>\min \{d_{{\it ff}}(h_1,\dots ,h_m),d_{{\it KSyz}}(h_1,\dots ,h_m)\}$, we have 
\[d_{{\it ff}}(h_1,\dots ,h_m)+1\leq  d_{\it reg}(h_1,\dots ,h_m).\]
\end{Thm}
\begin{proof}
If $H_1(K(h_1,\dots, h_m)_\bullet )_{\leq d}=0$ for any $1\leq d\leq d_{\it reg}(h_1,\dots ,h_m)-1$, then the system $h_1,\dots, h_m$ is regular up to degree $d_{\it reg}(h_1,\dots ,h_m)-1$ by Proposition \ref{prop:Die15}. Namely, the system $h_1,\dots ,h_m$ is semi-regular which contradicts the assumption. Hence, there exists $1\leq d\leq d_{\it reg}(h_1,\dots ,h_m)-1$ such that $H_1(K(h_1,\dots, h_m)_\bullet )_{\leq d}\not =0$ and we have $d_{{\it KSyz}}(h_1,\dots ,h_m)\leq  d_{\it reg}(h_1,\dots ,h_m)-1$.
Furthermore, we assume that $d_0\geq 1$ and $q>\min \{d_{{\it ff}}(h_1,\dots ,h_m),d_{{\it KSyz}}(h_1,$ $\dots ,h_m)\}$.
Then, since $d_{{\it ff}}(h_1,\dots ,h_m)=d_{{\it KSyz}}(h_1,\dots ,h_m)$ by Proposition \ref{prop:3}, we obtain $d_{{\it ff}}(h_1,\dots ,h_m)\leq d_{\it reg}(h_1,\dots ,h_m)-1$.
\end{proof}
More generally, we have a numerical upper bound of the first fall degree over a sufficiently large field. 
\begin{Def}\label{def:20230325a}
For a homogeneous system $h_1,\dots ,h_m\in \mathbb{F}_q[x_1,\dots, x_n]$, we put $\sum _da_dt^d=\prod _{i=1}^{m}(1-t^{d_i})/ (1-t)^n$ and define $D_{\mathbb{Z}_{\geq 0}}(h_1,\dots, h_m)=\inf \{d\mid a_d<0\}\cup \{\infty \}$. 
\end{Def}
\begin{Thm}\label{thm:20230221a}
For a homogeneous system $h_1,\dots, h_m\in \mathbb{F}_q[x_1,\dots, x_n]$, we have
\[d_{{\it KSyz}}(h_1,\dots ,h_m)\leq  D_{\mathbb{Z}_{\geq 0}}(h_1,\dots ,h_m).\] 
Moreover, if $\deg h_i=d_0\geq 1$ and $q>\min \{d_{{\it ff}}(h_1,\dots ,h_m),d_{{\it KSyz}}(h_1,\dots ,h_m)\}$, we have 
\[d_{{\it ff}}(h_1,\dots ,h_m)\leq  D_{\mathbb{Z}_{\geq 0}}(h_1,\dots ,h_m).\]
\end{Thm}
\begin{proof}
Since it is obviously for $D_{\mathbb{Z}_{\geq 0}}=\infty $, we may assume that $d_0:=D_{\mathbb{Z}_{\geq 0}}\in \mathbb{Z}_{\geq 0}$. By Proposition \ref{prop:Die15}, we have $H_1(K(h_1,\dots, h_m)_\bullet )_{\leq d_0}\not =0$. It follows that $d_{{\it KSyz}}(h_1,\dots ,h_m)\leq d_0$. 
Moreover, we have $d_{{\it ff}}(h_1,\dots ,h_m)=d_{{\it KSyz}}(h_1,\dots ,h_m)$ under the assumption $d_0\geq 1$ and $q>\min \{d_{{\it ff}}(h_1,\dots ,h_m),d_{{\it KSyz}}(h_1,\dots ,h_m)\}$. Thus, $d_{{\it ff}}(h_1,\dots ,h_m)\leq  D_{\mathbb{Z}_{\geq 0}}(h_1,\dots ,h_m)$. 
\end{proof}
\begin{Rem}
{\rm For any homogeneous system $h_1,\dots ,h_m$, we have $D_{\it reg}(h_1,\dots, h_m)\leq D_{\mathbb{Z}_{\geq 0}}(h_1,\dots ,h_m)$}.
\end{Rem}
\begin{Cor}\label{cor:20230325b}
Suppose $D_{\it reg}(h_1,\dots, h_m)=D_{\mathbb{Z}_{\geq 0}}(h_1,\dots ,$ $h_m)$. 
For a homogeneous system $h_1,\dots, h_m\in \mathbb{F}_q[x_1,\dots, x_n]$, we have
\begin{equation}\label{eq:20230222a}
d_{{\it KSyz}}(h_1,\dots ,h_m)\leq d_{\it reg}(h_1,\dots ,h_m).
\end{equation}
Equality holds in \eqref{eq:20230222a} if and only if $h_1,\dots, h_m$ is semi-regular. 
Furthermore, if $\deg h_i=d_0\geq 1$ and $q>\min \{d_{{\it ff}}(h_1,\dots ,h_m),d_{{\it KSyz}}(h_1,\dots ,h_m)\}$, we have 
\begin{equation}\label{eq:20230222b}
d_{{\it ff}}(h_1,\dots ,h_m)\leq  d_{\it reg}(h_1,\dots ,h_m).
\end{equation}
Equality holds in \eqref{eq:20230222b} if and only if $h_1,\dots, h_m$ is semi-regular. 
\end{Cor}
\begin{proof}
If the system $h_1,\dots ,h_m$ is not semi-regular, then \eqref{eq:20230222a} holds by Theorem \ref{thm:20230222c}. 
We assume that the system $h_1,\dots, h_m$ is semi-regular. Then, we have $d_{\it KSyz}(h_1,\dots, h_m)\geq d_{\it reg}(h_1,\dots, h_m)$ and $d_{\it reg}(h_1,\dots, h_m)=D_{\it reg}(h_1,\dots, h_m)$. Moreover, by Theorem \ref{thm:20230221a}, $d_{\it KSyz}(h_1,\dots, h_m)\leq D_{\mathbb{Z}_{\geq 0}}(h_1,\dots, h_m)=D_{\it reg}(h_1,\dots, h_m)=d_{\it reg}(h_1,\dots, h_m)$.
Hence, we obtain $d_{{\it KSyz}}(h_1,\dots ,h_m)= d_{\it reg}(h_1,\dots ,h_m)$. 
Conversely, if $d_{{\it KSyz}}(h_1,\dots ,h_m)= d_{\it reg}(h_1,\dots ,h_m)$ holds, then $H_1(K(h_1,\dots, h_m)_\bullet )_d=0$ for every $1\leq d\leq d_{\it reg}(h_1,\dots ,h_m)-1$. By Proposition \ref{prop:Die15}, the system $h_1,\dots ,h_m$ is regular up to degree $d_{\it reg}(h_1,\dots ,h_m)-1$ and namely is semi-regular.
Since $d_{{\it ff}}(h_1,\dots ,h_m)=d_{{\it KSyz}}(h_1,\dots ,h_m)$ holds under the assumption $d_0\geq 1$ and $q>\min \{d_{{\it ff}}(h_1,\dots ,h_m),d_{{\it KSyz}}(h_1,\dots ,h_m)\}$, the remainder of the assertion also holds. 
\end{proof}
\begin{Rem}
{\rm By Lemma \ref{lem:2}, the condition $q> D_{\mathbb{Z}_{\geq 0}}(h_1,\dots, h_m)$ depending on the degrees and the number of variables gives that the assumption $q>\min \{d_{{\it ff}}(h_1,\dots ,h_m),d_{{\it KSyz}}(h_1,\dots ,h_m)\}$ holds. 
}
\end{Rem}

\begin{Rem}
In this article, the assumption $q>\min \{d_{{\it ff}}(h_1,\dots ,h_m),d_{{\it KSyz}}(h_1,$ $\dots ,h_m)\}$ is only used to derive that $d_{{\it ff}}(h_1,\dots ,h_m)=d_{{\it KSyz}}(h_1,\dots ,h_m)$. All our assertions about upper bounds of the first fall degree hold for any polynomial system satisfying $d_{{\it ff}}(h_1,\dots ,h_m)=d_{{\it KSyz}}(h_1,\dots ,h_m)$.
\end{Rem}

\section{Estimation using a multi-grading}\label{sec:multi}
In this section, we introduce a value using a multi-degree to approximate the solving degree and, by extending Proposition \ref{prop:Die15} to a multi-grading, prove that the first fall degree of a multi-graded polynomial system over a sufficiently large field is bounded by this value. 
Moreover, we provide the theoretical assumption for applying the XL algorithm with a kernel search to a multi-graded polynomial system.

In this article, put $S=\mathbb{F}_q[{\bf x}_1,\dots, {\bf x}_s]$ and define $n_i=\sharp {\bf x}_i$. For ${\bf d}=(d_1,\dots ,d_s)\in \mathbb{Z}_{\geq 0}^s$ and variables $t_1,\dots ,t_{s}$, denote ${\bf t}^{{\bf d}}=t_1^{d_1}\cdots t_s^{d_s}$ and $|{\bf d}|=d_1+\cdots +d_s$. For an order on $\mathbb{Z}_{\geq 0}^s$, the symbol $\infty $ is defined as ${\bf d}\prec \infty $ for any ${\bf d}\in \mathbb{Z}_{\geq 0}^s$.

\subsection{Extend Proposition \ref{prop:Die15} to a multi-grading}\label{ssec:extmulti}

In this subsection, we show that a value in the following definition gives an upper bound for the first fall degree of a multi-graded polynomial system.

\begin{Def}\label{def:9}
Let $S$ be the $\mathbb{Z}_{\geq 0}^s$-graded polynomial ring.
For $\mathbb{Z}_{\geq 0}^s$-homogeneous polynomials $h_1,\dots ,h_m$ in $S$, we put 
\begin{equation}\label{eq:mltiHS}
\sum _{{\bf d}\in \mathbb{Z}_{\geq 0}^s}a_{{\bf d}}{\bf t}^{{\bf d}}=\prod _{i=1}^m(1-{\bf t}^{\deg h_i}){\rm HS}_S({\bf t}),
\end{equation}
 and define $D_{\mathbb{Z}_{\geq 0}^s}(h_1,\dots ,h_m):=\inf \{|{\bf d}|: a_{{\bf d}}<0\}\cup \{\infty \}.$
Moreover, for a well-ordering $\prec $ on $\mathbb{Z}_{\geq 0}^s$, we define ${\bf D}_{\mathbb{Z}_{\geq 0}^s,\prec}(h_1,\dots ,h_m):=\inf_{\prec }\{{\bf d}\mid a_{{\bf d}}<0\}\cup \{\infty \}$.
\end{Def}

The value $D_{\mathbb{Z}_{\geq 0}^s}$ in Definition \ref{def:9} is similar to $D_{{\it mgd}}$ in \cite{Nak20b}, but our value is also available for a wighted degree and is a more general concept.
Moreover, we extend Definition \ref{def:12} to the following:
\begin{Def}\label{def:123}
For $\mathbb{Z}_{\geq 0}^s$-homogeneous polynomial $h_1,\dots ,$ $h_m$ in the $\mathbb{Z}_{\geq 0}^s$-graded polynomial ring $S$ with a well-ordering $\prec $ on $\mathbb{Z}_{\geq 0}^s$, we define ${\bf d}_{{\it KSyz},\prec}$ as 
\[{\bf d}_{{\it KSyz},\prec }(h_1,\dots ,h_m)=\inf _{\prec } \{{\bf d}\mid {\rm Syz}(h_1,\dots ,h_m)_{{\bf d}}\not ={\rm KSyz}(h_1,\dots ,h_m)_{{\bf d}}\}\cup \{\infty \}.\]
If $s=1$, we denote this by $d_{{\it KSyz},\prec }$. 
\end{Def}

\noindent For the standard grading, i.e. $\deg _{\mathbb{Z}_{\geq 0}}x_i=1$, the value $d_{{\it KSyz},<}$ coincides with $d_{{\it KSyz}}$ in Definition \ref{def:12}.
We can extend Definition \ref{def:regupto} and Proposition \ref{prop:Die15} to a multi-grading as follows.

\begin{Def}\label{def:7}
Let $S$ be the $\mathbb{Z}_{\geq 0}^s$-graded polynomial ring and $\preceq $ be a well-ordering on $\mathbb{Z}_{\geq 0}^s$.
Then, $\mathbb{Z}_{\geq 0}^s$-homogeneous system $h_1,\dots ,h_m$ is {\rm regular up to degree {\bf d}} if the following multiplication map $\times h_i$ by $h_i$ is injective for each $i=1,\dots ,m$.
\[\hspace{-15pt}\times h_i:(S/\langle h_1,\dots ,h_{i-1}\rangle )_{{\bf d}_0-\deg h_i}\rightarrow (S/\langle h_1,\dots ,h_{i-1}\rangle )_{{\bf d}_0} \quad (\forall {\bf d}_0\preceq {\bf d}).\]
\end{Def}

\begin{Lem}\label{main}
Let $S$ be the $\mathbb{Z}_{\geq 0}^s$-graded polynomial ring and $\prec $ be a well-ordering on $\mathbb{Z}_{\geq 0}^s$ compatible with $< $ on $\mathbb{Z}_{\geq 0}$ such that ${\bf a}\prec {\bf b}$ if $|{\bf a}|<|{\bf b}|, {\bf a},{\bf b}\in \mathbb{Z}_{\geq 0}^s$.
For $\mathbb{Z}^s_{\geq 0}$-homogeneous system $h_1,\dots ,h_m$ with ${\bf d}_i:=\deg h_i\not ={\bf 0}$, the following conditions are equivalent:
\begin{enumerate}
\item $h_1,\dots ,h_m$ is regular up to degree ${\bf d}$
\item ${\rm HS}_{S/\langle h_1,\dots ,h_m\rangle }({\bf t})\equiv_{\preceq {\bf d}}\prod _{i=1}^m(1-{\bf t}^{\deg h_i}) {\rm HS}_S({\bf t})$
\item $H_1(K(h_1,\dots ,h_m)_\bullet )_{\preceq {\bf d}}=\{0\}$
\end{enumerate}
\end{Lem}
\begin{proof}
The proof proceeds in the same way as Proposition \ref{prop:Die15} (\cite{Die15}).
For the assertion 1$\;\Leftrightarrow \;$2, we suppose that the assertion holds for $m-1$. 
Then, for any ${\bf d}_0\preceq {\bf d}$, by the injection 
\[\hspace{-15pt}\times h_m :(S/\langle h_1,\dots ,h_{m-1}\rangle )_{{\bf d}_0-{\bf d}_m}\rightarrow (S/\langle h_1,\dots ,h_{m-1}\rangle )_{{\bf d}_0},\]
we have ${\rm HS}_{S/\langle h_1,\dots ,h_m\rangle }({\bf t})={\rm HS}_R({\bf t})-{\rm HS}_{\langle h_m\rangle _{R}}({\bf t})\equiv _{\preceq {\bf d}}{\rm HS}_R({\bf t})-{\rm HS}_{R}({\bf t})\cdot {\bf t}^{\deg h_m}=(1-{\bf t}^{\deg h_m}){\rm HS}_{R}({\bf t})$ where $R=S/\langle h_1,\dots ,h_{m-1}\rangle $.
Conversely, the map $\times h_m$ is injective if ${\rm HS}_{S/\langle h_1,\dots ,h_m\rangle }({\bf t})\equiv _{\preceq {\bf d}}(1-{\bf t}^{\deg h_m}){\rm HS}_{R}({\bf t})$ holds.

For the assertion 1$\;\Rightarrow \;$3, we suppose that the assertion holds for $m-1$. 
For any ${\bf d}_0\preceq {\bf d}$, the long exact sequence \eqref{eq:les} induces an exact sequence 
\[\hspace{-15pt}\cdots \rightarrow H_1(K(h_1,\dots ,h_{m-1})_\bullet )_{{\bf d}_0}\rightarrow H_1(K(h_1,\dots ,h_{m})_\bullet )_{{\bf d}_0}\]
\[\hspace{-15pt}\rightarrow (S/\langle h_1,\dots ,h_{m-1}\rangle )(-{\bf d}_m)_{{\bf d}_0}\overset{\times h_m}{\rightarrow }(S/\langle h_1,\dots ,h_{m-1}\rangle )_{{\bf d}_0}.\]
Then, by an assumption of the induction and the injection $\times h_m$, we have \[H_1(K(h_1,\dots ,h_m)_\bullet )_{{\bf d}_0}=0.\]

For the assertion 3$\;\Rightarrow \;$1, it suffices to prove the following statement and the case $l=m-1$ in particularly:
\[H_1(K(h_1,\dots ,h_m))_{\preceq {\bf d}}=0\Rightarrow  H_1(K(h_1,\dots ,h_l))_{\preceq {\bf d}}=0, 1\leq \forall l\leq m.\] 
Indeed, by the long exact sequence 
\[\hspace{-10pt}\cdots \rightarrow H_1(K(h_1,\dots ,h_l)_\bullet )\rightarrow (S/\langle h_1,\dots ,h_{l-1}\rangle )(-{\bf d}_l)\overset{\times h_l}{\rightarrow }S/\langle h_1,\dots ,h_{l-1}\rangle , \]
the right hand condition for each $1\leq l\leq m$ gives the injection 
\[\hspace{-10pt}\times h_l :(S/\langle h_1,\dots ,h_{l-1}\rangle )(-{\bf d}_l)_{{\bf d}_0}\overset{\times h_l}{\rightarrow }(S/\langle h_1,\dots ,h_{l-1}\rangle)_{ {\bf d}_0}\quad ( \forall {\bf d}_0\preceq {\bf d}).\]
Suppose that there exists the minimum ${\bf d}'$ such that $H_1(K(h_1,\dots ,h_{m-1})_\bullet )_{{\bf d}'}\not =0$.
Then, by $|{\bf d}_m|>0$ and ${\bf d}'\succ {\bf d}'-{\bf d}_m$, 
\[H_1(K(h_1,\dots ,h_{m-1})_\bullet )(-{\bf d}_m)_{{\bf d}'}\]\[=H_1(K(h_1,\dots ,h_{m-1})_\bullet )_{{\bf d}'-{\bf d}_m}=0.\]
Hence we have $H_1(K(h_1,\dots ,h_m)_\bullet )_{{\bf d}'}\not =0$ by the short exact sequence
\[\hspace{-20pt}H_1(K(h_1,\dots ,h_{m-1})_\bullet )(-{\bf d}_m)_{{\bf d}'}\overset{\times h_m}{\rightarrow }H_1(K(h_1,\dots ,h_{m-1})_\bullet )_{{\bf d}'}\]\[\rightarrow H_1(K(h_1,\dots ,h_{m})_\bullet )_{{\bf d}'},\]
in the long exact sequence \eqref{eq:les}.
Since $H_1(K(h_1,\dots ,h_m)_\bullet )_{\preceq {\bf d}}$ $=0$, we have ${\bf d}\prec {\bf d}'$.
Therefore $H_1(K(h_1,\dots ,h_{m-1})_\bullet )_{\preceq {\bf d}}=0$.
\end{proof}

By this lemma, we see that the values in Definition \ref{def:7} are an upper bound for the first fall degree of a multi-graded polynomial system.

\begin{Thm}\label{thm:8}
Let $S$ be a $\mathbb{Z}_{\geq 0}^s$-graded polynomial ring. Then $S$ is $\mathbb{Z}_{\geq 0}$-graded with $S_d=\oplus _{d=|{\bf d}|}S_{{\bf d}}$. 
For $\mathbb{Z}_{\geq 0}^s$-homogeneous polynomials $h_1,\dots ,h_m$ with $\deg _{\mathbb{Z}_{\geq 0}^s}h_i\not ={\bf 0}$, we have 
\[d_{{\it KSyz}}(h_1,\dots ,h_m)\leq  D_{\mathbb{Z}_{\geq 0}^s}(h_1,\dots ,h_m).\]
Moreover, if $S=\mathbb{F}[{\bf x}_1,\dots ,{\bf x}_s]$ with $\deg _{\mathbb{Z}_{\geq 0}^s}x_{ij}={\bf e}_i$, $\deg h_i=d_0\geq 2$ and $q>\min \{d_{{\it ff}}(h_1,\dots ,h_m),d_{{\it KSyz}}(h_1,\dots ,h_m)\}$, we have  
\[d_{{\it ff}}(h_1,\dots ,h_m)\leq D_{\mathbb{Z}_{\geq 0}^s}(h_1,\dots ,h_m).\]
\end{Thm}
\begin{proof}
Since the assertion is on $S_d=\oplus _{d=|{\bf d}|}S_{{\bf d}}$, we may fix a well-ordering $\prec $ on $\mathbb{Z}_{\geq 0}^s$ in Lemma \ref{main} as the graded lexicographic monomial ordering.
When $D_{\mathbb{Z}_{\geq 0}^s}=\infty$, the statement is obviously. 
If $D_{\mathbb{Z}_{\geq 0}^s}<\infty$, the power series \eqref{eq:mltiHS} has a negative coefficient at a certain ${\bf d}$ such that $|{\bf d}|=D_{\mathbb{Z}_{\geq 0}^s}(h_1,\dots ,h_m)$. 
Then, by Lemma \ref{main} and the positivity of the coefficients in the Hilbert series, there exists a non-Koszul syzygy of $\mathbb{Z}_{\geq 0}^s$-degree equal to or less than ${\bf d}$ with respect to $\prec$.
Thus $d_{{\it KSyz},<}\leq  D_{\mathbb{Z}_{\geq 0}^s}(h_1,\dots ,h_m)$ on $\mathbb{Z}_{\geq 0}$-graded $S$ with $S_d=\oplus _{d=|{\bf d}|}S_{\bf d}$.
Assume that $S=\mathbb{F}[{\bf x}_1,\dots ,{\bf x}_s]$ and it is $\mathbb{Z}_{\geq 0}^s$-graded by $\deg _{\mathbb{Z}_{\geq 0}^s}x_{ij}={\bf e}_i$.
For the standard grading, when $d_0\geq 2$ and $q>\min \{d_{{\it ff}}(h_1,\dots ,h_m),d_{{\it KSyz}}(h_1,\dots ,h_m)\}$, we have the last assertion since $d_{{\it ff}}(h_1,\dots ,h_m)=d_{{\it KSyz},<}(h_1,\dots ,h_m)$ by Lemma \ref{prop:3}.
\end{proof}

The following theorem is a simple generalization of Theorem \ref{thm:8}, but it has an application (see \ref{ssec:app}).
\begin{Thm}\label{thm:8gen}
Let $S$ be the $\mathbb{Z}_{\geq 0}^s$-graded polynomial ring and $\prec $ be a well-ordering on $\mathbb{Z}_{\geq 0}^s$ compatible with $<$ on $\mathbb{Z}_{\geq 0}$ such that ${\bf a}\prec {\bf b}$ if $|{\bf a}|<|{\bf b}|,{\bf a},{\bf b}\in \mathbb{Z}_{\geq 0}^s$.
Then, for $\mathbb{Z}_{\geq 0}^s$-homogeneous polynomials $h_1,\dots ,h_m$, we have 
\[{\bf d}_{{\it KSyz},\prec }(h_1,\dots ,h_m)\;\preceq \;{\bf D}_{\mathbb{Z}_{\geq 0}^s,\prec}(h_1,\dots ,h_m).\]
\end{Thm}
\begin{proof}
Since it is obviously for ${\bf D}_{\mathbb{Z}_{\geq 0}^s,\prec}=\infty$, we may assume that ${\bf d}_0:={\bf D}_{\mathbb{Z}_{\geq 0}^s,\prec}(h_1,\dots ,h_m)\in \mathbb{Z}_{\geq 0}^s$.
By Lemma \ref{main} with the same well-ordering as the statement, we have $H_1(K(h_1,\dots ,h_m)_\bullet )_{\preceq {\bf d}_0}\not =\{0\}$.
It follows that ${\bf d}_{{\it KSyz},\prec }(h_1,\dots ,h_m)\;\preceq {\bf d}_0$.
\end{proof}

\subsection{Theoretical assumption for the XL algorithm with a kernel search}\label{ssec:estXL}

As mentioned in Subsection \ref{ssec:solverks}, the XL algorithm generates the Macaulay matrix of degree $d$ and its complexity is estimated as \eqref{eq:cmplxlks}. 
For the generated Macaulay matrix, the value $c_d$ in \eqref{eq:cmplxlks} is the number of columns minus the rank. If $c_d\leq 1$ holds, a kernel vector obtained by a kernel search always corresponds to a solution to the system.
In this subsection, we consider a theoretical assumption for controlling the value $c_d$.

The number of columns of the Macaulay matrix is that of the monomials of degree $d$, and the row space is isomorphic to the vector space $\<f_1,\dots ,f_m\>_d$ over $\mathbb{F}$.
Hence, the number of columns is $\dim _\mathbb{F}\mathbb{F}[x_1,\dots ,x_n]_d$, and the rank is $\dim _\mathbb{F}\<f_1,\dots ,f_m\>_d$.
Namely, \[c_d=\dim _\mathbb{F}\mathbb{F}[x_1,\dots ,x_n]_d-\dim _\mathbb{F}\<f_1,\dots ,f_m\>_d.\]
The tool \eqref{eq:degree0322} is rewritten as 
\[\hspace{-10pt}d^{\leq 1}=\min \{d>0\mid \dim _\mathbb{F}\mathbb{F}[x_1,\dots ,x_n]_d-\dim _\mathbb{F}\<f_1,\dots ,f_m\>_d\leq 1\}.\]
Since $\dim _\mathbb{F}\mathbb{F}[x_1,\dots ,x_n]_d-\dim _\mathbb{F}\<f_1,\dots ,f_m\>_d =\dim _\mathbb{F}(\mathbb{F}[x_1,$ $\dots ,x_n]_d/\<f_1,\dots ,f_m\>_d)$, we need to consider a case in which it is possible to combinatorially compute the Hilbert series ${\rm HS}_{\mathbb{F}[x_1,\dots ,x_n]/\<f_1,\dots ,f_m\>}(t)$.

For the dimension $u_d$ explained in Subsection \ref{ssec:solvermgd}, if $u_d>1$, then $c_d>1$.
Hence we have to generally consider an upper bound for $\dim _\mathbb{F}\mathbb{F}[x_1,\dots ,x_n]_d-\dim _\mathbb{F}\<f_1,\dots ,f_m\>_d$.
Moreover, we discuss a more general statement applied to the multi-graded case (see Subsection \ref{ssec:solvermgd}). 
Thus we define the following tool as a generalization of \eqref{eq:degree0322}:
\begin{Def}\label{def:20250914e}
For $\mathbb{Z}_{\geq 0}^s$-homogeneous polynomials $f_1,\dots ,$ $f_m$, an integer $a$ and a well-ordering $\prec $ on $\mathcal{X}\subseteq \mathbb{Z}_{\geq 0}^s$, we define 
\begin{equation}\label{eq:degree0322b}
{\bf d}_{\mathcal{X},\prec}^{\leq a}(f_1,\dots ,f_m):=\min _\prec \{{\bf d}\in \mathcal{X}\mid \dim \mathbb{F}[{\bf x}_1,\dots ,{\bf x}_s]_{{\bf d}}-\dim \<f_1,\dots ,f_m\>_{{\bf d}}\leq a\}.
\end{equation}
In particular, we denote it by ${\bf d}_{\mathbb{Z}_{\geq a}^s,\prec}^{\leq a}$ when $\mathcal{X}=\mathbb{Z}_{\geq a}^s$. Moreover, we define $d_{\mathcal{X},\prec}^{\leq a}(f_1,\dots ,f_m)=|{\bf d}_{\mathcal{X},\prec}^{\leq a}(f_1,\dots ,f_m)|$.
\end{Def}

As shown in Proposition \ref{prop:Die15} and Lemma \ref{main}, the regularity for the polynomial system provides the partial information of the Hilbert series, and we generalize Definition \ref{def:7} as follows:
\begin{Def}\label{def:0322d}
Let $f_1,\dots ,f_m \in \mathbb{F}[{\bf x}_1,\dots ,{\bf x}_s]$ be a $\mathbb{Z}_{\geq 0}^s$-graded system.
For a set $\mathcal{X}\subseteq \mathbb{Z}_{\geq 0}^s$, the system $f_1,\dots ,f_m$ is {\it regular on $\mathcal{X}$} if, for any $0\leq i\leq m-1$ and ${\bf d}\in \mathcal{X}$, the following map is injective:
\[\times f_{i+1}:(\mathbb{F}[{\bf x}_1,\dots ,{\bf x}_s]/\<f_1,\dots ,f_i\>)_{{\bf d}-\deg f_{i+1}}\rightarrow (\mathbb{F}[{\bf x}_1,\dots ,{\bf x}_s]/\<f_1,\dots ,f_i\>)_{{\bf d}}.\]
\end{Def}

Note that, for a $\mathbb{Z}_{\geq 0}^s$-graded module $M$, we set $M_{{\bf d}}=0$ if ${\bf d}$ has a negative component.
The following lemma is a generalization of the partial statement $``1\Rightarrow 2"$ in Lemma \ref{main}:

\begin{Lem}
Let ${\bf d}=(d_1,\dots ,d_s)\in \mathbb{Z}_{\geq 0}^s$ and $\mathcal{X}_{\bf d}=\{(i_1,\dots ,i_s)\mid i_j\leq d_j\;(\forall j)\}$.
If a $\mathbb{Z}_{\geq 0}^s$-graded system $f_1,\dots ,f_m \in \mathbb{F}[{\bf x}_1,\dots ,{\bf x}_s]$ is regular on $\mathcal{X}_{\bf d}$, for every $(i_1,\dots ,i_s)\in \mathcal{X}_{\bf d}$, the coefficient of $t_1^{i_1}\cdots t_s^{i_s}$ in ${\rm HS}_{\mathbb{F}[{\bf x}_1,\dots ,{\bf x}_s]/\<f_1,\dots ,f_m\>}({\bf t})$ agrees with that in 
\[{\rm HS}_{\mathbb{F}[{\bf x}_1,\dots ,{\bf x}_s]}({\bf t})\cdot \prod _{i=1}^m(1-{\bf t}^{\deg f_i}),\]
where ${\bf t}=(t_1,\dots ,t_s)$. 
\end{Lem}

Proposition \ref{prop:0322c} shows that the degree ${\bf d}_{\mathbb{Z}_{\geq a}^s,\prec }^{\leq a}(f_1,\dots ,$ $f_m)$ in Definition \ref{eq:degree0322b} is bounded by the following value under the regularity introduced in Definition \ref{def:0322d}.

\begin{Def}\label{def:20250914a}
For a $\mathbb{Z}_{\geq 0}^s$-graded system $f_1,\dots ,f_m \in \mathbb{F}[{\bf x}_1,\dots ,{\bf x}_s]$ and the well-order $\prec $ on $\mathcal{X}\subseteq \mathbb{Z}_{\geq 0}^s$, we define ${\bf D}_{\mathcal{X},\prec }^{\leq a}(f_1,\dots ,f_m)$ as follows:
Put $\{a_{\bf d}\}_{{\bf  d}\in \mathbb{Z}_{\geq 0}^s}$ as 
$\sum _{{\bf d}\in \mathbb{Z}_{\geq 0}^s}a_{{\bf d}}{\bf t}^{{\bf  d}}={\rm HS}_{\mathbb{F}[{\bf x}_1,\dots ,{\bf x}_s]}({\bf t})\cdot \prod _{i=1}^m(1-{\bf t}^{\deg f_i})$ and define 
\[{\bf D}_{\mathcal{X},\prec }^{\leq a}(f_1,\dots ,f_m)=\min  _\prec \{{\bf d}\in \mathcal{X}\mid a_{{\bf d}}\leq a\},\]
\[{D}_{\mathcal{X},\prec }^{\leq a}(f_1,\dots ,f_m)=|{\bf D}_{\mathcal{X},\prec }^{\leq a}(f_1,\dots ,f_m)|.\]
In particular, we denote it by ${\bf D}_{\mathbb{Z}_{\geq a}^s,\prec }^{\leq a}$ when $\mathcal{X}=\mathbb{Z}_{\geq a}^s$.
\end{Def}

\begin{Prop}\label{prop:0322c}
Let $\prec $ be a well-ordering on $\mathbb{Z}_{\geq 0}^s$. 
If a $\mathbb{Z}_{\geq 0}^s$-graded system $f_1,\dots ,f_m \in \mathbb{F}[{\bf x}_1,\dots ,{\bf x}_s]$ is regular on $\mathcal{X}_{\bf d}$ for every ${\bf d}\prec {\bf d}_{\mathbb{Z}_{\geq a}^s,\prec }^{\leq a}(f_1,\dots ,f_m)$, then
\[{\bf d}_{\mathbb{Z}_{\geq a}^s,\prec }^{\leq a}(f_1,\dots ,f_m)\preceq {\bf D}_{\mathbb{Z}_{\geq a}^s,\prec }^{\leq a}(f_1,\dots ,f_m).\]
\end{Prop}

\begin{Rem}
{\rm 
Note that ${\bf D}_{\mathbb{Z}_{\geq 0}^s,\prec }={\bf D}_{\mathbb{Z}_{\geq 0}^s,\prec }^{\leq a}$ if $a=-1$, $D_{\mathbb{Z}_{\geq 0}^s}=|{\bf D}_{\mathbb{Z}_{\geq 0}^s,\prec }^{\leq a}|$ if $a=-1$ and $\prec $ is compatible with the standard degree, $D_{{\it reg}}={\bf D}_{\mathbb{Z}_{\geq 0}^s,\prec }^{\leq a}$ if $a=0$ and $s=1$.
For the system arisen from an attack in multivariate cryptography, a non-positive coefficient in the multivariate power series \eqref{eq:220325a} is often negative, namely ${\bf D}_{\mathbb{Z}_{\geq 0}^s,\prec }^{\leq 0}={\bf D}_{\mathbb{Z}_{\geq 0}^s,\prec }^{\leq -1}$ often holds. }
\end{Rem}

\begin{Prop}\label{cor:20230417a}
Let $f_1,\dots ,f_m $ be homogeneous polynomials in $S=\mathbb{F}[x_1,\dots, x_n]$ such that $\deg f_1=\cdots =\deg f_m\geq 1$.
If $f_1,\dots ,f_m $ is regular on $\{d\mid  0\leq d< D_{\mathbb{Z}_{\geq 0}}(f_1,\dots ,f_m) \}$ and 
$q> D_{\mathbb{Z}_{\geq 0}}(f_1,\dots ,f_m)$, 
then we have
\[d_{\it ff}(f_1,\dots ,f_m)=D_{\mathbb{Z}_{\geq 0}}(f_1,\dots ,f_m).\]
\end{Prop}
\begin{proof}
Since $q> D_{\mathbb{Z}_{\geq 0}}(f_1,\dots ,f_m)$, we have $d_{\it ff}(f_1,\dots ,f_m)$ $=d_{\it KSyz}(f_1,\dots ,f_m)\leq D_{\mathbb{Z}_{\geq 0}}(f_1,\dots ,f_m)$ by Theorem \ref{thm:20230221a}. 
Moreover, for $0\leq d<D_{\mathbb{Z}_{\geq 0}}(f_1,\dots ,f_m)$, we have $H_1(K(f_1,\dots ,f_m)_\bullet )_d=0$ by Lemma \ref{main}. 
Hence $d_{\it ff}(f_1,\dots ,f_m)=d_{\it KSyz}(f_1,\dots, f_m)\geq D_{\mathbb{Z}_{\geq 0}}(f_1,\dots ,f_m)$.
\end{proof}

\subsection{Comparison to the previous work in the bilinear case}\label{ssec:20250914b}
In this subsection, we explain the relation between Baena, Cabarcas, and Verbel's work \cite{BCV20} on bilinear systems and our paper.
Namely, we consider homogeneous polynomials $F=\{f_1,\dots, f_m\} \subset \mathbb{F}[{\bf x},{\bf y}]$, $\deg_{\mathbb{Z}_{\geq 0}^2} f_i=(1,1)$, $n_x=\sharp {\bf x}$, and $n_y=\sharp {\bf y}$. 
Here we follow the notation of \cite{BCV20}.
In \cite{BCV20}, they studied Procedure \ref{proc:20250913a} with ${\bf d}=(1,d-1)$ as an algorithm for solving bilinear systems.
They introduced the degree of regularity and the first fall degree on the bi-degree $(1,d-1)$ and determined them precisely under some regularity assumptions. In particular, the first fall degree studied in \cite{BCV20} can be seen as a special case of Definition \ref{def:12}, and it can be written as follows:
\begin{Def}[\cite{BCV20}]
$d_{{\bf y}\text{-}{\it ff}}(F):=\min \{d \mid {\rm KSyz}(F)\cap \mathbb{F}[{\bf y}]_{d-2}^m \not ={\rm Syz}(F)\cap \mathbb{F}[{\bf y}]_{d-2}^m\}$
\end{Def}
Let ${\bm M}_{{\bf y},d}(F)$ be the degree-$d$ part of the Macaulay matrix of Procedure \ref{proc:20250913a} with ${\bf d}=(1,d-1)$, and $J_{\bf x}(F)=(\partial f_i/\partial x_j)_{1\leq i\leq m,1\leq j\leq n_x}$.
Since $\mathbb{F}[{\bf y}]^m\cap {\rm KSyz}(F)=\{{\bm 0}\}$, it is important, by the following lemma, to study the injectivity of the map $J_{\bf x}(F)_d: \mathbb{F}[{\bf y}]_{d-2}^m \rightarrow \mathbb{F}[{\bf y}]_{d-1}^{n_x},{\bm g}\mapsto {\bm g}\cdot J_{\bf x}(F)$ in order to investigate $d_{{\bf y}\text{-}{\it ff}}(F)$.
\begin{Lem}[\cite{BCV20}]\label{lem:20250913d}
${\rm Syz}(F)\cap \mathbb{F}[{\bf y}]_{d-2}^m={\rm Ker}(J_{\bf x}(F)_d)$.
\end{Lem}
At this point, they introduced the notions of regularity, degree of regularity, and regular sequence with respect to the degrees $(1,d-1)$ as follows.
\begin{Def}[\cite{BCV20}] 
Let $F\subset \mathbb{F}[{\bf x},{\bf y}]$ be a set of bi-homogeneous polynomials.
\begin{enumerate}
\item $F$ is said to be ${\bf y}$-$d$-regular if ${\rm rank}({\bm M}_{{\bf y},d}(F))=m\binom{n_y+d-3}{d-2}$.
\item $d_{{\bf y},{\it reg}}(F):=\min \{d \mid {\rm rank}({\bm M}_{{\bf y},d}(F))=n_x\binom{n_y+d-2}{d-1}\}$.
\item $F$ is said to be ${\bf y}$-semiregular if $F$ is ${\bf y}$-$d$-regular for all $d<d_{{\bf y},{\it reg}}(F)$.
\end{enumerate}
\end{Def}
As a result, they determined the degree of regularity and the first fall degree with respect to the bi-degree $(1,d-1)$, and clarified their relation:
\begin{Prop}[\cite{BCV20}]
For a ${\bf y}$-semiregular sequence $F$, the following hold:
\begin{enumerate}
\item $d_{{\bf y},{\it reg}}(F)=\left \lceil \frac{n_x(n_y-1)}{m-n_x}\right \rceil +1$
\item $d_{{\bf y}\text{-}{\it ff}}(F)=\min \{d>0 \mid m\binom{n_y+d-3}{d-2}>n_x\binom{n_y+d-2}{d-1}\}$
\end{enumerate}
In particular, if $(m-n_x)\mid n_x(n_y-1)$, then $d_{{\bf y}\text{-}{\it ff}}(F)=d_{{\bf y},{\it reg}}(F)+1$, otherwise $d_{{\bf y}\text{-}{\it ff}}(F)=d_{{\bf y},{\it reg}}(F)$.
\end{Prop}
\begin{Rem}[\cite{BCV20}]
The map $J_{\bf x}(F)_d$ corresponds, as a linear map over $\mathbb{F}$, to an $m\binom{n_y+d-3}{d-2}\times n_x\binom{n_y+d-2}{d-1}$ matrix.
Thus if $m\binom{n_y+d-3}{d-2}>n_x\binom{n_y+d-2}{d-1}$, then ${\rm Syz}(F)\cap \mathbb{F}[{\bf y}]^m_{d-2}\not =\{{\bm 0}\}$.
\end{Rem}
Essentially, our paper does not achieve such precise results as theirs, but we explain a partial relation below.
In particular, we rewrite part of their results in terms of the formal power series we introduced.
For this purpose, we consider the value $D_{\mathcal{X},\prec }^{\leq a}$ in Definition \ref{def:20250914a} on bi-degrees $(1,d-1)$ as follows:
\begin{Rem}\label{rem:20250912a}
{\rm i)} We consider the lexicographic order $\prec _{lex}$ as the well-ordering on $\{1\}\times \mathbb{Z}_{\geq 0}$, i.e. $(1,d_1)\prec_{lex} (1,d_2)\Leftrightarrow d_1<d_2$.  
When setting $\mathcal{X}=\{1\}\times \mathbb{Z}_{\geq 0}$ and $\prec \;=\; \prec _{lex}$, the value $D_{\mathcal{X},\prec}^{\leq a}$ in Definition \ref{def:20250914a} is given as follows:
\[\hspace{-10pt}D_{\mathcal{X},\prec}^{\leq a}=\min \{d \mid {\rm Coeff}(S_{(m,n_x,n_y)}(t_x,t_y),t_xt_y^{d-1})\leq a\},\]
where $S_{(m,n_x,n_y)}(t_x,t_y)=(1-t_xt_y)^m(1-t_x)^{-n_x}(1-t_y)^{-n_y}$.
\noindent {\rm ii)} For the formal power series $S_{(m,n_x,n_y)}(t_x,t_y)$, we have\\
\[
\begin{array}{l}
{\rm Coeff}\left (S_{(m,n_x,n_y)}(t_x,t_y),t_xt_y^{d-1}\right )\\
={\rm Coeff}\left ((1-t_y)^{-n_y}(1+n_xt_x)(1-mt_xt_y),t_xt_y^{d-1}\right )\\
={\rm Coeff}\left ((1-t_y)^{-n_y}(n_xt_x-mt_xt_y),t_xt_y^{d-1}\right )\\
=n_x\binom{n_y+d-2}{d-1}-m\binom{n_y+d-3}{d-2}.
\end{array}
\]
\end{Rem}
By Remark \ref{rem:20250912a}, ii), their notion of regularity can be restated such as Proposition \ref{prop:Die15} and Lemma \ref{main}:
\begin{Cor}
Let $F=\{f_1,\dots, f_m\}\subset \mathbb{F}[{\bf x},{\bf y}]$ with $\deg _{\mathbb{Z}_{\geq 0}^2}f_i=(1,1)$, $n_x=\sharp {\bf x}$, and $n_y=\sharp {\bf y}$, and $S_{(m,n_x,n_y)}(t_x,t_y)$ be the notation in Remark \ref{rem:20250912a}.
The following are equivalent:
\begin{enumerate}
\item $F$ is ${\bf y}$-$d$-regular
\item ${\rm HS}_{\mathbb{F}[{\bf x},{\bf y}]/\langle F\rangle }(t_x,t_y)\equiv _{(1,d-1)}S_{(m,n_x,n_y)}(t_x,t_y)$
\item $H_1(K(f_1,\dots, f_m)_\bullet )_{(1,d-1)}=0$
\end{enumerate}
\end{Cor}
\begin{proof}
Since $\dim _\mathbb{F}\langle F\rangle _{(1,d-1)}={\rm rank}({\bm M}_{{\bf y},d}(F))$ and $\dim _\mathbb{F}\mathbb{F}[{\bf x},{\bf y}]_{(1,d-1)}=n_x\binom{n_y+d-2}{d-1}$, Remark \ref{rem:20250912a} implies $``1\Leftrightarrow 2"$. Moreover, since ${\rm rank}({\bm M}_{{\bf y},d}(F))=\dim_\mathbb{F} {\rm Im}(J_{\bf x}(F))$, we have ${\rm rank}({\bm M}_{{\bf y},d}(F))=m\binom{n_y+d-3}{d-2}=\dim _\mathbb{F}\mathbb{F}[{\bf y}]_{d-2}^m\Leftrightarrow {\rm Ker}(J_{\bf x}(F))=\{{\bm 0}\}$.
Since Lemma \ref{lem:20250913d} and $H_1(K(f_1,\dots, f_m)_\bullet )_{(1,d-1)}={\rm Syz}(F)\cap \mathbb{F}[{\bf y}]^m_{d-2}$, it follows that $``1\Leftrightarrow 3"$.
\end{proof}
Therefore, we have the following facts:
\begin{Cor}\label{cor:20250913b}
Let $F=\{f_1,\dots, f_m\}\subset \mathbb{F}[{\bf x},{\bf y}]$ with $\deg _{\mathbb{Z}_{\geq 0}^2}f_i=(1,1)$, $n_x=\sharp {\bf x}$, and $n_y=\sharp {\bf y}$, and $D_{\{1\}\times \mathbb{Z}_{\geq 0},\prec _{lex}}^{\leq -1}$ be the notation in Remark \ref{rem:20250912a}.
Then, $d_{{\bf y}\text{-}{\it ff}}(F)\leq D_{\{1\}\times \mathbb{Z}_{\geq 0},\prec _{lex}}^{\leq -1}$.
\end{Cor}
Note that $d_{{\bf y},{\it reg}}(F)=d_{\{1\}\times \mathbb{Z}_{\geq 0},\prec _{lex}}^{\leq 0}(F)$ (see Definition \ref{def:20250914e}) and $d_{{\bf y},{\it reg}}(F)\leq D_{\{1\}\times \mathbb{Z}_{\geq 0},\prec _{lex}}^{\leq 0}$ under the assumption of Proposition \ref{prop:0322c}. However, in the bi-degree $(1,d-1)$ setting, the following holds:
\begin{Cor}\label{cor:20250914c}
Let $F$ be a ${\bf y}$-semiregular sequence, $D_{\mathbb{Z}_{\geq 0}^2}$ be the notation in Definition \ref{def:9} with $s=2$, and $D_{\{1\}\times \mathbb{Z}_{\geq 0},\prec _{lex}}^{\leq 0}$ and $D_{\{1\}\times \mathbb{Z}_{\geq 0},\prec _{lex}}^{\leq -1}$ be the notation in Remark \ref{rem:20250912a}.
Then, the following hold:
\begin{enumerate}
\item $d_{{\bf y},{\it reg}}(F)=D_{\{1\}\times \mathbb{Z}_{\geq 0},\prec _{lex}}^{\leq 0}$
\item $d_{{\bf y}\text{-}{\it ff}}(F)=D_{\{1\}\times \mathbb{Z}_{\geq 0},\prec _{lex}}^{\leq -1}$
\item $D_{\mathbb{Z}^2_{\geq 0}}\leq d_{{\bf y}\text{-}{\it ff}}(F)$
\end{enumerate}
\end{Cor}
From the above results, $D_{\{1\}\times \mathbb{Z}_{\geq 0},\prec _{lex}}^{\leq 0}$ is an appropriate object for Procedure \ref{proc:20250913a} with ${\bf d}=(1,d-1)$ when the input system is ${\bf y}$-semiregular. 
On the other hand, the value $D_{\mathbb{Z}^2_{\geq 0}}$ we introduced takes into account the bi-degree in $D_{\{1\}\times \mathbb{Z}_{\geq 0},\prec _{lex}}^{\leq 0}$, and when $D_{\mathbb{Z}^2_{\geq 0}}<D_{\{1\}\times \mathbb{Z}_{\geq 0},\prec _{lex}}^{\leq 0}$, this suggests the possibility of more efficient solving in a different bi-degree from the form $(1,d-1)$. However, the actual efficiency depends on the algorithm and the input parameters, so a detailed analysis is required.
\begin{Rem}
If $(m-n_x)\mid n_x(n_y-1)$, then there exists a positive integer $d$ such that $m\binom{n_y+d-3}{d-2}=n_x\binom{n_y+d-2}{d-1}$.
By their result in \cite{BCV20} and Corollary \ref{cor:20250913b}, in this case $F$ is also ${\bf y}$-$d_{{\bf y},{\it reg}}$-regular.
That is, if $F$ is ${\bf y}$-semiregular, then for all $d< D_{\{1\}\times \mathbb{Z}_{\geq 0},\prec _{lex}}^{\leq -1}$, the system $F$ is ${\bf y}$-$d$-regular.
\end{Rem}

\section{Application to multivariate cryptography}\label{ssec:app}
In Subsection \ref{ssec:extmulti}, we give an upper bound for the first fall degree $d_{{\it ff}}$ of a polynomial system $f_1,\dots ,f_m$ whose top homogeneous component is $\mathbb{Z}_{\geq 0}^s$-homogeneous under the assumption that $\deg f_i=d_0\geq 2$ and 
\begin{equation}\label{eq:assumption}
q>\min \{d_{{\it ff}}(f_1,\dots ,f_m),d_{{\it KSyz}}(f_1,\dots ,f_m)\}.
\end{equation}
In this section, we show that this assumption is satisfied in actual attacks against multivariate public-key signature schemes Rainbow \cite{NIST2R} and G$e$MSS \cite{NIST2RGeMSS} proposed in NIST PQC standardization project \cite{NIST2016}. 

In Subsection \ref{ssec:kra}, we recall the multivariate public-key signature scheme. 
In Subsection \ref{ssec:rbsattack} (resp. Subsection \ref{ssec:minrankattackks}), we explain the RBS attack (resp.  the MinRank attack) as a key recovery attack against Rainbow (resp. G$e$MSS), and show that the quadratic system generated by the attack satisfies the condition \eqref{eq:assumption}. 

\subsection{Key recovery attack}\label{ssec:kra}
For simplicity, we treat only homogeneous polynomials as polynomials and assume that the characteristic of the field $\mathbb{F}$ is odd.
Then, a quadratic homogeneous polynomial in $\mathbb{F}[x_1,\dots ,x_n]$ corresponds to a symmetric $n\times n$ matrix over $\mathbb{F}$.
A polynomial system $(f_1,\dots ,f_m)$ of $\mathbb{F}[x_1,\dots ,x_n]^m$ gives a map $\mathbb{F}^n\rightarrow \mathbb{F}^m$ by $ {\bf a}\mapsto (f_1({\bf a}),\dots ,f_m({\bf a}))$ which is called a {\it polynomial map}.  
For a set $\mathcal{C}$ of easily invertible quadratic maps, a multivariate public-key signature scheme consists of the following three algorithms:
\begin{itemize}
\item[] \underline{Key generation}: We randomly construct two invertible linear maps $U:\mathbb{F}^n\rightarrow \mathbb{F}^n$ and $ T:\mathbb{F}^m\rightarrow \mathbb{F}^m$, and $F:\mathbb{F}^n\rightarrow \mathbb{F}^m \in \mathcal{C}$ which is called a {\it central map}, and then compute the composition $P:=T\circ F\circ U.$ The {\it public key} is given as $P$. The tuple $(T,F,U)$ is a {\it secret key}.
\item[] \underline{Signature generation}: For a message ${\bf b}\in \mathbb{F}^m$, we compute ${\bf b}'=T^{-1}({\bf b})$.
Next, we can compute an element ${\bf a}'$ of $F^{-1}(\{{\bf b}'\})$ since $F$ is easily invertible. Consequently, we obtain a signature ${\bf a}=U^{-1}({\bf a}')\in \mathbb{F}^n$.
\item[] \underline{Verification}: We verify whether $P({\bf a})={\bf b}$ holds.
\end{itemize}
For a given public key $P$, the {\it key recovery attack} generates two invertible linear maps $U',T'$ and $F'\in \mathcal{C}$ such that $P=T'\circ F'\circ U'$ and forges a signature for any message.
For matrices $M_{p_i}$, $M_{f_i'}$, $M_{U'}$ and $M_{T'}$ corresponding to $p_i,f_i',U'$ and $T'$ where $P=(p_1,\dots ,p_m)$ and $F=(f_1',\dots ,f_m')$, respectively, we have
\[\hspace{-20pt}(M_{p_1},\dots ,M_{p_m})=(M_{U'}M_{f_1'}{}^tM_{U'},\dots ,M_{U'}M_{f_m'}{}^tM_{U'})M_{T'}.\]
The RBS attack \cite{DYCCC} takes time to generate a part of $M_{U'}$ and $M_{T'}$ of UOV \cite{KPG99} or its mulit-layerization, i.e.\! Rainbow \cite{DS05}.
For a public key $P=(p_1,\dots ,p_m)$ of HFE and a central $n\times n$ matrix $M_c=(c_{ij})_{i,j}$ over $\mathbb{F}_{q^n}$, we have
\[(M_{p_1},\dots ,M_{p_m})=(M_{U'}M_{\varphi }M_{1}{}^tM_{\varphi }{}^tM_{U'},\dots ,M_{U'}M_{\varphi }M_{m}{}^tM_{\varphi }{}^tM_{U'})M_{\varphi }^{-1}M_{T'},\]
where $M_{\varphi }=(\theta ^{(i-1)q^{j-1}})_{1\leq i,j\leq n}$, $\mathbb{F}_{q^n}=\mathbb{F}_q[\theta ]$ and $M_k=(c_{i-k,j-k}^{q^k})$.
The MinRank attack with the KS method takes time to generate a column of $M_{\varphi }^{-1}M_{T'}$ of HFE or its modifications, e.g.\! G$e$MSS.

\subsection{The RBS attack}\label{ssec:rbsattack}

The RBS attack is an attack generating a secret key of Rainbow. 
Let $v,o_1$ and $o_2$ be Rainbow parameters.
For $n\times n$ matrices $M_{p_1},\dots ,M_{p_m}$ corresponding to a public quadratic system $(p_1,\dots ,p_m)$ where $n=v+o_1+o_2$ and $m=o_1+o_2$, the RBS dominant system is a quadratic system in $\mathbb{F}[x_1,\dots ,x_{v+o_1},y_1,\dots ,y_{o_2}]^{m+n-1}$ consisting of
\[\hspace{-15pt}(x_1,\dots ,x_{v+o_1},0,\dots ,0,1)M_{p_i}{}^t(x_1,\dots ,x_{v+o_1},0,\dots ,0,1), 1\leq i\leq m,\]
and the first $n-1$ components of
\[\hspace{-10pt}(x_1,\dots ,x_{v+o_1},0,\dots ,0,1)M_{p_1}+\sum _{j=1}^{o_2}y_j(x_1,\dots ,x_{v+o_1},0,\dots ,0,1)M_{p_{o_1+j}}.\] 
Since $(x_1,\dots ,x_{v+o_1},0,\dots ,0,1)$ and ${}^t(1,0,\dots ,0,y_1,\dots ,$ $y_{o_2})$ correspond to a row and a column of secret linear transformations, the RBS attack can generate a part of the secret key by solving the system.
Since the polynomial ring $S=\mathbb{F}_q[x_1,\dots ,x_{v+o_1},y_1,\dots ,y_{o_2}]$ is $\mathbb{Z}_{\geq 0}^2$-graded by $\deg _{\mathbb{Z}_{\geq 0}^2}x_i=(1,0)$ and $\deg _{\mathbb{Z}_{\geq 0}^2}y_i=(0,1)$, the top homogeneous component of the RBS dominant system is contained in $S_{(1,1)}^{n-1}\oplus S_{(2,0)}^{m}$ and is $\mathbb{Z}_{\geq 0}^2$-homogeneous.
Then, for the $\mathbb{Z}_{\geq 0}^2$-graded polynomial ring $S$, the power series \eqref{eq:mltiHS} in Definition \ref{def:9} is 
\begin{equation}\label{eq:RBS}
\frac{(1-t_1t_2)^{v+o_1+o_2-1}(1-t_1^2)^{o_1+o_2}}{(1-t_1)^{v+o_1}(1-t_2)^{o_2}}.
\end{equation}
The paper \cite{Nak20} experimentally shows that the solving degree of the RBS dominant system is tightly approximated by $D_{\mathbb{Z}_{\geq 0}^2}$ in Definition \ref{def:9} which is written as $D_{{\it bgd}}$ in \cite{Nak20}.
According to \cite{Nak20}, for the Rainbow parameters Ia and IIIc/Vc \cite{NIST2R} proposed in NIST PQC 2nd round, the best complexities of the Rainbow-Based-Separation attack are given by $D_{\mathbb{Z}_{\geq 0}^2}=15$ and $23/30$, respectively.
Since $q=16$ and $256$ for the parameter Ia and IIIc/Vc, it follows that $q>D_{\mathbb{Z}_{\geq 0}^2}$ holds.
In particular, since $q>d_{{\it KSyz}}$ by Theorem \ref{thm:8}, the assumption \eqref{eq:assumption} holds and the second half of Theorem \ref{thm:8} holds.
Namely, the value $D_{\mathbb{Z}_{\geq 0}^2}$ in the paper \cite{Nak20} gives an upper bound for the first fall degree $d_{{\it ff}}$.
Furthermore, Smith-Tone and Perlner \cite{SP20} propose an XL algorithm as the bi-graded version of that explained in Subsection \ref{ssec:solvermgd} and provide the complexity estimation for the attack by using a certain theoretical value.
Then we can use Theorem \ref{thm:8gen} as a theoretical background of the estimation.
Here note that their theoretical value in \cite{SP20} is defined by a non-positive coefficient appeared in \eqref{eq:RBS}, namely, it is different from our value ${\bf D}_{\mathbb{Z}_{\geq 0}^2}$, and requires another theoretical background based on such as a conjecture in Diem \cite{Die04}.

\subsection{The MinRank attack using the KS method}\label{ssec:minrankattackks}

Multivariate signature scheme G$e$MSS \cite{NIST2RGeMSS} is a minus and vinegar modification of HFE \cite{Pat96}.
The MinRank attack with the KS method \cite{KS98} is an attack generating a secret key of a multivariate cryptosystem such as G$e$MSS and Rainbow.
Although a public quadratic system of G$e$MSS is defined over the field $\mathbb{F}_2$ of order two, the complexity of the attack is dominated by that of a Gr\"{o}bner basis algorithm for solving a certain system over a very large field, say the KS system. 
Let $n,D,a$ and $v$ be G$e$MSS parameters.
For $(n+v)\times (n+v)$ matrices $M_{p_1},\dots ,M_{p_{n-a}}$ over $\mathbb{F}_2$ corresponding to the public quadratic system $(p_1,\dots ,p_{n-a})$, the MinRank attack finds $x_1,\dots ,x_{n-a}$ in $\mathbb{F}_{2^n}$ such that 
\[{\rm rank}\left (\sum _{i=1}^{n-a} x_iM_{p_i}\right )\leq r,\]
where $r=\lceil \log _2(D-1)\rceil +a+v$.
Then, since a found vector $(x_1,\dots ,x_{n-a})$ corresponds to a column vector of a certain linear transform over $\mathbb{F}_{2^n}$, the MinRank attack can generate a part of a secret key.
For finding $x_1,\dots ,x_{n-a}$, the Kinis-Shamir modeling solves the KS system in $\mathbb{F}_{2^n}$ which is a quadratic system in $\mathbb{F}_2[{\bf x},{\bf k}_1,\dots ,{\bf k}_c]^{c(n-a)}$ and is the components of 
\[(0,\dots ,0,\overset{j}{1},0,\dots ,0,k_{j1},\dots ,k_{jr})\left (\sum _{i=1}^{n-a}x_iM_{p_i}\right ),\]\[ 1\leq j\leq c,\]
where ${\bf x}=\{x_1,\dots ,x_{n-a}\}$, ${\bf k}_j=\{k_{j1},\dots ,k_{jr}\}$ and $c\leq n-a-r$. 
Since the polynomial ring $S=\mathbb{F}_2[{\bf x},{\bf k}_1,\dots ,{\bf k}_c]$ is $\mathbb{Z}_{\geq 0}^{c+1}$-graded by 
\[\deg _{\mathbb{Z}_{\geq 0}^{c+1}}x_i=(1,0,\dots ,0),\text{ and }\deg _{\mathbb{Z}_{\geq 0}^{c+1}}k_{jl}=(0,\dots ,0,\overset{j+1}{1},0,\dots ,0),\] 
the top homogeneous component of the KS system is contained in $\bigoplus _{j=2}^{c+1}S_{{\bf e}_1+{\bf e}_{j}}^{n-a}$ and is $\mathbb{Z}_{\geq 0}^{c+1}$-homogeneous where ${\bf e}_j=(0,\dots ,0,\overset{j}{1},0,\dots ,0)$.
Then, for the $\mathbb{Z}_{\geq 0}^{c+1}$-graded polynomial ring $S$, the power series \eqref{eq:mltiHS} in Definition \ref{def:9} is 
\begin{equation}\label{eq:KS}\frac{(1-t_0t_1)^{n-a}\cdots (1-t_0t_c)^{n-a}}{(1-t_0)^m(1-t_1)^r\cdots (1-t_c)^r}.\end{equation}
The paper \cite{Nak20b} experimentally shows that the solving degree of the KS system from the Minrank problem \cite{Fau08} is approximated by $D_{\mathbb{Z}_{\geq 0}^{c+1}}$ in Definition \ref{def:9} which is written as $D_{{\it mdg}}$ in \cite{Nak20b}.
Then they also show that the value $D_{\mathbb{Z}_{\geq 0}^{c+1}}$ which is smaller than the order $q=256$, i.e. the assumption \eqref{eq:assumption} holds, improves the complexity of the MinRank attack with the KS method against Rainbow. 
In particular, the value $D_{\mathbb{Z}_{\geq 0}^{c+1}}$ in \cite{Nak20b} gives an upper bound for the first fall degree $d_{{\it ff}}$.
Meanwhile, the G$e$MSS parameter sets for a security of $2^{128}$, $2^{192}$ and $2^{256}$ proposed in NIST PQC 2nd round take $n\approx 174, 265$ and $354$ \cite{NIST2RGeMSS}, respectively.
Thus, for these proposed parameter sets, the order $q$ in the definition of the first fall degree $d_{{\it ff}}$ is around $ 2^{174}, 2^{265}$ and $2^{354}$, respectively, and the complexities at $c=1$ of the KS method are given by $D_{\mathbb{Z}_{\geq 0}^{c+1}}=26, 44$ and $65$, respectively.
It follows that $q\gg D_{\mathbb{Z}_{\geq 0}^{c+1}}$ holds.
In particular, $q>d_{{\it KSyz}}$ by Theorem \ref{thm:8}, and the second half of Theorem \ref{thm:8} holds.
Therefore, this $D_{\mathbb{Z}_{\geq 0}^{c+1}}$ gives an upper bound for the first fall degree $d_{{\it ff}}$.

\section*{Acknowledgement}
This work was supported by JST CREST Grant Number JPMJCR2113, and JSPS KAKENHI Grant Number JP20K19802 and JP23K16885. The authors are grateful to Yacheng Wang and Yasuhiko Ikematsu for their comments on the article.





\end{document}